\ExplSyntaxOn
\AtBeginDocument
 {
  \NewCommandCopy{\originalbibitem}{\bibitem}
  \cs_set_eq:NN \bibitem \bibitem:w
 }

\cs_set_protected:Npn \bibitem:w #1 \par
 {
  \tl_set:Nn \l_tmpa_tl {#1}
  \regex_replace_all:nnN { ([A-Z])\s } { \1\.\  } \l_tmpa_tl
  \exp_last_unbraced:NV \originalbibitem \l_tmpa_tl \par
 }
\ExplSyntaxOff

\documentclass[a4paper,UKenglish,cleveref, autoref, thm-restate]{lipics-v2021}
\nolinenumbers
\makeatletter
\renewcommand\normalsize{\@setfontsize\normalsize{11}{12}}
\makeatother

\newcommand{\procOrFullVersion}[2]{#2}

\procOrFullVersion{\long\gdef\onlyProc#1{#1}}{\long\gdef\onlyProc#1{}}
\procOrFullVersion{\long\gdef\onlyFull#1{}}{\long\gdef\onlyFull#1{#1}}

\usepackage[utf8]{inputenc} 
\usepackage[T1]{fontenc}    
\usepackage{lmodern}
\usepackage{amsmath,amsthm,amssymb}
\newtheorem{fact}[equation]{Fact}

\usepackage[linesnumbered]{algorithm2e}
\usepackage{hyperref}       
\usepackage{url}            
\usepackage{booktabs}       
\usepackage{amsfonts}       
\usepackage{nicefrac}       
\usepackage{microtype}      
\usepackage{cleveref}       
\usepackage{lipsum}         
\usepackage{graphicx}

\usepackage[numbers]{natbib}
\usepackage{doi}
\usepackage{stackengine}

\usepackage{xfrac}
\usepackage{stmaryrd}
\usepackage{multirow}
\usepackage{thm-restate}

\usepackage{setspace}

\usepackage{slashbox}

\usepackage{wrapfig}

\swapnumbers

\swapnumbers

\usepackage{tikz}
\usetikzlibrary{graphs, shapes, arrows.meta, calc, positioning}
\pgfdeclarelayer{bg}    
\pgfsetlayers{bg, main}

\definecolor{earthyellow}{rgb}{0.88, 0.66, 0.37}
\definecolor{bostonuniversityred}{rgb}{0.8, 0.0, 0.0}
\definecolor{persianblue}{rgb}{0.11, 0.22, 0.73}

\pgfdeclarelayer{bg}    
\pgfsetlayers{bg, main}

\newcommand{\PLang}[2][]{
\def\bannachempty{#1}
\ifx\empty\bannachempty\text{p-}\else\text{p$_{#1}$-}\fi
\textsc{#2}
}

\newcommand{\R}{\mathbb{R}}

\newcommand{\IN}{\mathbb{N}}
\newcommand{\IR}{\mathbb{R}}

\usepackage{etoolbox}

\makeatletter
\newcommand{\DeclareMathActive}[2]{
  
  \expandafter\edef\csname keep@#1@code\endcsname{\mathchar\the\mathcode`#1 }
  \begingroup\lccode`~=`#1\relax
  \lowercase{\endgroup\def~}{#2}
  \AtBeginDocument{\mathcode`#1="8000 }
}

\newcommand{\compactEquals}[1]{\let\IsInPP=1#1\let\IsInPP=\relax}
\newcommand{\PP}[1]{\mathbb{P}(\compactEquals{#1})}

\newcommand{\pp}[1]{{P}(\compactEquals{#1})}

\newsavebox{\neqbox}
\savebox{\neqbox}{$\neq$}
\def\neqinPP{\mathrel{\usebox{\neqbox}}}

\newcommand{\existsR}{\ensuremath{\mathsf{\exists\IR}}}
\newcommand{\ETR}{\ensuremath{\textsc{ETR}}}
\newcommand{\succETR}{\ensuremath{\textup{succ}\text{-}\textsc{ETR}}}

\newcommand{\succR}{\ensuremath{\mathsf{succ\text{-}\exists\IR}}}

\newcommand{\NP}{\ensuremath{\mathsf{NP}}}
\newcommand{\cP}{\ensuremath{\mathsf{P}}}
\newcommand{\NEXP}{\ensuremath{\mathsf{NEXP}}}

\newcommand{\realNEXP}{\ensuremath{\mathsf{NEXP}_{\text{real}}}}
\newcommand{\PSPACE}{\ensuremath{\mathsf{PSPACE}}}
\newcommand{\EXPSPACE}{\ensuremath{\mathsf{EXPSPACE}}}

\newcommand{\ccPP}{\ensuremath{\mathsf{PP}}}
\newcommand{\sharpP}{\#\mathsf{P}}

\newcommand{\leqp}{\leq_{\text{P}}}

\DeclareMathOperator{\poly}{poly}

\newcommand{\maxvaluecount}{c}
\newcommand{\cL}{{\mathcal L}}
\newcommand{\cE}{{\mathcal E}}
\newcommand{\cF}{{\mathcal F}}
\newcommand{\cT}{{\mathcal T}}

\newcommand{\fM}{{\mathfrak M}}

\newcommand{\bQ}{{\bf Q}}

\newcommand{\bU}{{\bf U}}
\newcommand{\bV}{{\bf V}}
\newcommand{\bW}{{\bf W}}
\newcommand{\bY}{{\bf Y}}
\newcommand{\bX}{{\bf X}}

\newcommand{\cB}{{\cal B}}
\newcommand{\cC}{{\cal C}}
\newcommand{\cG}{{\cal G}}

\newcommand{\bc}{{\bf c}}

\newcommand{\be}{{\bf e}}
\newcommand{\bff}{{\bf f}}

\newcommand{\bi}{{\bf i}}

\newcommand{\bp}{{\bf p}}
\newcommand{\bq}{{\bf q}}

\newcommand{\bt}{{\bf t}}
\newcommand{\bu}{{\bf u}}
\newcommand{\bv}{{\bf v}}
\newcommand{\bw}{{\bf w}}
\newcommand{\by}{{\bf y}}
\newcommand{\bx}{{\bf x}}
\newcommand{\bz}{{\bf z}}

\newcommand{\pa}{{\bf pa}}
\newcommand{\Pa}{{\bf Pa}}

\def\Eprop{\cE_{\textit{prop}}}
\def\Eint{\cE_{\textit{int}}}
\def\Epint{\cE_{\textit{post-int}}}
\def\Ecounter{\cE_{\textit{counterfact}}}

\def\probname#1#2{^{\text{#1}}_{\text{#2}}}

\def\probsumname#1#2{^{\text{#1}{\langle{\scriptscriptstyle\Sigma}\rangle}}_{\text{#2}}}

\def\probnamesum#1#2{^{\text{#1}{\langle{\scriptscriptstyle\Sigma}\rangle}}_{\text{#2}}}

\def\probsumgraphname#1#2{^{\text{#1}{\langle{\scriptscriptstyle\Sigma}\rangle}}_{\mbox{\scriptsize\sc DAG}, \text{#2}}}

\def\Tprobpolysum{\cT\probnamesum{poly}{1}}

\def\Tprobpoly{\cT\probname{poly}{1}}

\def\Tcausalpolysum{\cT\probnamesum{poly}{3}}

\def\Tipolysum{\cT\probnamesum{poly}{$i$}}

\def\Tilinsum{\cT\probsumname{lin}{$i$}}
\def\Ticompsum{\cT\probsumname{base}{$i$}}
\def\Tipoly{\cT\probname{poly}{$i$}}

\def\Tilin{\cT\probname{lin}{$i$}}
\def\Ticomp{\cT\probname{base}{$i$}}

\def\Lprobstar{\cL^{*}_{1}} 
 
\def\Lcausalstar{\cL^{*}_{3}}

\def\Listar{\cL^{*}_{i}}

\def\Lcausalpolysum{\cL\probnamesum{poly}{3}}

\def\Lipolysum{\cL\probnamesum{poly}{$i$}}

\def\Lab{\textit{Lab}}

\newcommand{\SATprobstar}{\mbox{\sc Sat}^{*}_{{\cL_1}}}

\newcommand{\SATcausalstar}{\mbox{\sc Sat}^{*}_{\cL_3}}
\newcommand{\SATistar}{\mbox{\sc Sat}^{*}_{\cL_i}}

\newcommand{\SATprobpoly}{\mbox{\sc Sat}^{\textit{poly}}_{{\cL_1}}}

\newcommand{\SATipoly}{\mbox{\sc Sat}^{\textit{poly}}_{{\cL_i}}}

\newcommand{\SATistargraph}{\mbox{\sc Sat}^{*}_{\mbox{\scriptsize\sc DAG}, {\cL_i}}}

\newcommand{\SATprobstargraph}{\mbox{\sc Sat}^{*}_{\mbox{\scriptsize\sc DAG}, {\cL_1}}}

\newcommand{\SATcausalstargraph}{\mbox{\sc Sat}^{*}_{\mbox{\scriptsize\sc DAG}, {\cL_3}}}

\newcommand{\SATprobcompgraph}{\mbox{\sc Sat}^{\textit{base}}_{\mbox{\scriptsize\sc DAG}, {\cL_1}}}

\newcommand{\SATproblingraph}{\mbox{\sc Sat}^{\textit{lin}}_{\mbox{\scriptsize\sc DAG}, {\cL_1}}}

\newcommand{\SATprobcompsumgraph}{\mbox{\sc Sat}\probsumgraphname{base}{${\cal L}_1$}}
\newcommand{\SATproblinsumgraph}{\mbox{\sc Sat}\probsumgraphname{lin}{${\cal L}_1$}}

\newcommand{\SATprobpolysumgraph}{\mbox{\sc Sat}\probsumgraphname{poly}{${\cal L}_1$}}

\newcommand{\SATinterventcompsumgraph}{\mbox{\sc Sat}\probsumgraphname{base}{${\cal L}_2$}}
\newcommand{\SATinterventlinsumgraph}{\mbox{\sc Sat}\probsumgraphname{lin}{${\cal L}_2$}}

\newcommand{\SATinterventpolysumgraph}{\mbox{\sc Sat}\probsumgraphname{poly}{${\cal L}_2$}}

\newcommand{\SATcausalcompsumgraph}{\mbox{\sc Sat}\probsumgraphname{base}{${\cal L}_3$}}

\newcommand{\SATcausalpolysumgraph}{\mbox{\sc Sat}\probsumgraphname{poly}{${\cal L}_3$}}

\def\probsumBN#1{^{\text{#1}{\langle{\scriptscriptstyle\Sigma}\rangle}}}
\newcommand{\SATprobstarBN}{\mbox{\sc BN-MC}^{*}}

\newcommand{\SATprobpolysumBN}{\mbox{\sc BN-MC}\probsumBN{poly}}

\def\false{\textsc{false}}
\def\true{\textsc{true}}

\newcommand{\SATprobcompsum}{\mbox{\sc Sat}\probsumname{base}{${\cal L}_1$}}
\newcommand{\SATproblinsum}{\mbox{\sc Sat}\probsumname{lin}{${\cal L}_1$}}

\newcommand{\SATprobpolysum}{\mbox{\sc Sat}\probsumname{poly}{${\cal L}_1$}}

\newcommand{\SATinterventcompsum}{\mbox{\sc Sat}\probsumname{base}{${\cal L}_2$}}
\newcommand{\SATinterventlinsum}{\mbox{\sc Sat}\probsumname{lin}{${\cal L}_2$}}

\newcommand{\SATinterventpolysum}{\mbox{\sc Sat}\probsumname{poly}{${\cal L}_2$}}

\newcommand{\SATcausalcompsum}{\mbox{\sc Sat}\probsumname{base}{${\cal L}_3$}}

\newcommand{\SATcausalpolysum}{\mbox{\sc Sat}\probsumname{poly}{${\cal L}_3$}}

\newcommand{\SATprobpolysumsm}{\mbox{\sc Sat}\probsumname{poly}{\textit{sm},\,${\cal L}_1$}}
\newcommand{\SATinterventpolysumsm}{\mbox{\sc Sat}\probsumname{poly}{\textit{sm},\,${\cal L}_2$}}
\newcommand{\SATcausalpolysumsm}{\mbox{\sc Sat}\probsumname{poly}{\textit{sm},\,${\cal L}_3$}}
\newcommand{\SATipolysumsm}{\mbox{\sc Sat}\probsumname{poly}{\textit{sm},\,${\cal L}_i$}}

\newcommand{\SATproblinsumsm}{\mbox{\sc Sat}\probsumname{lin}{\textit{sm},\,${\cal L}_1$}}
\newcommand{\SATinterventlinsumsm}{\mbox{\sc Sat}\probsumname{lin}{\textit{sm},\,${\cal L}_2$}}
\newcommand{\SATcausallinsumsm}{\mbox{\sc Sat}\probsumname{lin}{\textit{sm},\,${\cal L}_3$}}

\newcommand{\SATprobcompsumsm}{\mbox{\sc Sat}\probsumname{basic}{\textit{sm},\,${\cal L}_1$}}
\newcommand{\SATinterventcompsumsm}{\mbox{\sc Sat}\probsumname{basic}{\textit{sm},\,${\cal L}_2$}}
\newcommand{\SATcausalcompsumsm}{\mbox{\sc Sat}\probsumname{basic}{\textit{sm},\,${\cal L}_3$}}

\newcommand{\PPB}[1]{P_{\cB}(\compactEquals{#1})}

\def\Plus{\texttt{+}}
\def\Minus{\texttt{-}}
\newcommand{\cJ}{{\mathcal J}}

\newcommand\myldots{\!\makebox[1em][c]{.\hfil.\hfil.}}  
\def\mmid{ | }
\def\dop{\textit{do}}

\makeatletter
\newcommand*{\indep}{
  \mathbin{
    \mathpalette{\@indep}{}
  }
}
\newcommand*{\nindep}{
  \mathbin{
    
    \mathpalette{\@indep}{/}
                               
  }
}
\newcommand*{\@indep}[2]{

  \sbox0{$#1\perp\m@th$}
  \sbox2{$#1=$}
  \sbox4{$#1\vcenter{}$}
  \rlap{\copy0}
  \dimen@=\dimexpr\ht2-\ht4-.2pt\relax

  \kern\dimen@
  \ifx\\#2\\
  \else
    \hbox to \wd2{\hss$#1#2\m@th$\hss}
    \kern-\wd2 
  \fi
  \kern\dimen@
  \copy0 
}

\newenvironment{compactitem}{\begin{itemize}}{\end{itemize}}

\title{Probabilistic and Causal Satisfiability: \texorpdfstring{\\}{} Constraining the Model}  

\ArticleNo{170}

\author{Markus Bl\"{a}ser}{Saarland University, Germany}{mblaeser@cs.uni-saarland.de}{
https://orcid.org/0000-0002-1750-9036}{}
\author{Julian D\"{o}rfler}{Saarland University, Germany}{jdoerfler@cs.uni-saarland.de}{https://orcid.org/0000-0002-0943-8282}{}
\author{Maciej Li\'{s}kiewicz}{University of Lübeck, Germany}{maciej.liskiewicz@uni-luebeck.de}{https://orcid.org/0000-0003-0059-5086}{}
\author{Benito van der Zander}{University of Lübeck, Germany}{b.vanderzander@uni-luebeck.de}{https://orcid.org/0000-0001-5957-4621}{Work supported by the Deutsche Forschungsgemeinschaft (DFG) grant 471183316 (ZA 1244/1-1).}

\authorrunning{M. Bl\"{a}ser, J. D\"{o}rfler, M.  Li\'{s}kiewicz, and B. van der Zander} 
\Copyright{Markus Bl\"{a}ser, Julian D\"{o}rfler, Maciej Li\'{s}kiewicz, and Benito van der Zander} 

\ccsdesc{Theory of computation~Complexity theory and logic}

\keywords{Existential theory of the real numbers, Computational complexity, Probabilistic logic, Structural Causal Models
} 
\category{Track B: Automata, Logic, Semantics, and Theory of Programming} 

\onlyProc{\relatedversiondetails[]{Full Version}{todo https://arxiv.org/abs/} }
\onlyFull{\relatedversion{LIPIcs, Volume 334, ICALP 2025 - Paper No 170}}

\begin{document}

\maketitle

\begin{abstract}
We study the complexity of satisfiability problems in probabilistic and causal reasoning. Given random variables $X_1, X_2,\ldots$ over finite domains, the basic terms are probabilities of propositional formulas over atomic events $X_i = x_i$, such as $\PP{X_1 = x_1}$ or $\PP{X_1 = x_1 \vee X_2 = x_2}$. The basic terms can be combined using addition (yielding linear terms) or multiplication (polynomial terms). The probabilistic satisfiability problem asks whether a joint probability distribution satisfies a Boolean combination of (in)equalities over such terms. Fagin et al.\ \cite{fagin1990logic} showed that for basic and linear terms, this problem is $\NP$-complete, making it no harder than Boolean satisfiability, while Mossé et al.\ \cite{ibeling2022mosse} proved that for polynomial terms, it is complete for the existential theory of the reals.

Pearl’s Causal Hierarchy (PCH) extends the probabilistic setting with interventional and counterfactual reasoning,  enriching the expressiveness of languages. However, Mossé et al.\ \cite{ibeling2022mosse} found that satisfiability complexity remains unchanged.  Van der Zander et al.\ \cite{zander2023ijcai} showed that introducing a marginalization operator to languages induces a significant increase in complexity.

We extend this line of work by adding two new dimensions to the problem by constraining the models. First, we fix the graph structure of the underlying structural causal model, motivated by settings like Pearl’s do-calculus, and give a nearly complete landscape across different arithmetics and PCH levels. Second, we study small models. While earlier work showed that satisfiable instances admit polynomial-size models, this is no longer guaranteed with compact marginalization. We characterize the complexities of satisfiability under small-model constraints across different settings.

\end{abstract}

\setcounter{footnote}{0}       

\section{Introduction}\label{sec:intro}

Reasoning about probability is essential in many research fields. In computer science, it plays a crucial role in analyzing probabilistic programs, understanding a program's behavior under probabilistic input assumptions, and handling uncertain information in expert systems.
In a seminal paper, Fagin et al.~\cite{fagin1990logic}, introduce a logic to reason about probabilities. Their aim was to formulate a calculus that allows to phrase statements like ``$\PP{X=x} = 1/2$'' or ``$\PP{X = x} < \PP{Y = y}$''. They study (among other things) the satisfiability problem for Boolean combinations of such terms, i.e., is there a probability distribution that satisfies all the terms. It turns out that the expressiveness of the underlying calculus has a great influence on the complexity of such satisfiability problems. For instance, in their original work, Fagin et al.\ investigated three types of terms in the (in)equalities: basic ones, like the one above, which consist of only single basic probabilities, linear combinations of basic probabilities, and polynomial expressions in basic probabilities. This choice parameterizes the satisfiability problem and it turns out that the type of equations has an impact on the complexity of the corresponding satisfiability problem.

We will study the satisfiability of such formulas from a multi-parametric view: there are several parameters in the underlying calculus, like the type of equations mentioned above, and by adjusting these parameters, we get a large family of satisfiability of varying complexities, depending on the choice of our parameters. There has been a lot of work along these lines, see e.g.\ \cite{fagin1990logic, ibeling2020probabilistic, ibeling2022mosse, ibeling2024probabilistic, zander2023ijcai, blaser2024existential, IBELING2024103339,doerflerICLR2025}, and the main contribution of this paper is that we finally (almost) complete the whole picture.

Another dimension is that we can enhance the basic probability terms. This is inspired by the causal theory in AI.
The development of the modern causal theory in AI and empirical sciences has greatly benefited from an influential 
structured approach to inference about causal phenomena, which is based on a reasoning 
hierarchy named ``Ladder of Causation'', also often referred to as ``Pearl’s Causal Hierarchy'' 
(PCH) (\cite{shpitser2008complete,Pearl2009,bareinboim2022pearl}, see also  \cite{pearl2018book}
for a gentle introduction to the topic). This three-level framework formalizes various types of reasoning 
that reflect the progressive sophistication of human thought regarding causation. It arises from a 
collection of causal mechanisms that model the ``ground truth'' of \emph{unobserved} nature
formalized within a Structural Causal Model (SCM). These mechanisms are then combined 
with three patterns of reasoning concerning \emph{observed} phenomena expressed at the 
corresponding layers of the hierarchy, known as \emph{probabilistic} (also called 
\emph{associational} in the AI literature), 
\emph{interventional},  and \emph{counterfactual} (for formal definitions of these concepts as well as for illustrative examples, see 
Section~\ref{sec:preliminaries}).

A basic term at the probabilistic (observational) layer is expressed as a common probability, 
such as\footnote{In our paper, we consider random variables over discrete, finite domains. By an event,
we mean a propositional formula over atomic events of the form $X=x$, such 
as $(X=x\wedge Y=y)$ or $(X=x \vee Y\neq y)$. Moreover, by $\PP{Y=y, X=x}$, etc., we mean, 
as usually, $\PP{X=x \wedge Y=y}$. Finally by using lowercase letters in $\PP{x, y}$, we abbreviate  $\PP{Y=y, X=x}$. }
$\PP{x, y}$. This may represent queries like ``How likely does a patient 
have both diabetes $(X=x)$ and high blood pressure $(Y=y)$?'' This layer corresponds precisely to the calculus developed by Fagin et al.\ mentioned above.
%
The interventional patterns extend the basic probability terms by allowing the use of Pearl's 
do-operator~\cite{Pearl2009} which models an experiment like a randomized controlled trial \cite{fisher1936design}. 
For instance, $\PP{[x]y}$ which\footnote{A common and popular 
notation for the post-interventional probability is $\PP{Y=y\mmid \compactEquals{\dop(X=x)}}$. 
In this paper, however,  we use the notation $\PP{[X=x]Y=y}$ since it is more convenient for our analysis.}, 
in general differs from $\PP{y\mmid x}$, allows to ask hypothetical questions such as, e.g., 
``How likely it is that a patient's headache will be cured $(Y=y)$ if he or she takes aspirin $(X=x)$?''. 
An example formula at this layer is $\PP{[x]y}=\sum_z \PP{y\mmid x,z}\PP{z}$
which estimates the causal effect of the intervention  $\dop(X = x)$ (all patients take aspirin) 
on outcome variable $Y = y$ (headache cure). It illustrates the use of the prominent back-door 
adjustment to eliminate the confounding effect of a factor represented by variable $Z$~\cite{Pearl2009}.
The basic terms at the highest level of the hierarchy enable us to formulate queries related to counterfactual 
situations. For example, $\PP{[X=x]Y=y \mmid (X=x',Y=y')}$  expresses the probability that, for instance, 
a patient who did not receive a vaccine $(X=x')$ and died $(Y=y')$ would have lived ($Y=y$) if he or she had been vaccinated ($X=x$).

The computational complexity aspects of reasoning about uncertainty in this framework 
have been the subject of intensive studies in the past decades. The research has resulted 
in a large number of significant achievements, especially in the case of probabilistic 
inference with the input probability distributions encoded by Bayesian 
networks~\cite{pearl1988probabilistic}. These problems are of great interest in AI and 
important results from this perspective 
include~\cite{cooper1990computational,dagum1993approximating,roth1996hardness,park2004complexity,koller2009probabilistic}.
However, despite intensive research, many fundamental problems in the field remain open.

The main interest of our research is focused on the precise characterization of the computational complexity of satisfiability problems (and their validity counterparts) for languages of all PCH layers, combined with increasing the expressiveness of (in)equalities by enabling the use of more 
complex operators. Our primary concern is the effect of constraining the model in these satisfiability problems, on the one hand by fixing the graph structure of the model, and on the other by bounding the model size to be polynomial. 

\subsection{Reasoning about probabilities: a multi-parametric view}

One starting point for our investigations is a pioneering paper work by Fagin et al.~\cite{fagin1990logic}, who introduce a logic to reason about probabilities. We are given a set of random variables over some fixed finite domain $D$, often $\{0,1\}$.  They consider formulas consisting of Boolean combinations of (in)equalities of \emph{basic} and \emph{linear} terms, like 
$\PP{(X = 0 \vee Y=1) \wedge (X = 0 \vee Y = 0)}\compactEquals{=1} \wedge(\PP{X = 0}\compactEquals{=0}\vee \PP{X = 0}\compactEquals{=1})  \wedge(\PP{Y = 0}\compactEquals{=0}\vee \PP{Y = 0}\compactEquals{=1})$ with binary variables $X,Y$. The authors provide a complete axiomatization for the used logic, which is essentially a formalization of  Nilsson’s probabilistic logic~\cite{nilsson1986probabilistic} and they show that the problem of deciding satisfiability is  $\NP$-complete. Thus, surprisingly,  the complexity is no worse than that of propositional logic. Fagin et al.~extend then the language to (in)equalities of \emph{polynomial} terms, with the goal of reasoning about conditional probabilities. They prove that the latter variant is $\NP$-hard and contained in $\PSPACE$. Recently, Moss{\'e} et al.~\cite{ibeling2022mosse} have given the exact complexity of this problem, showing that deciding the satisfiability is~$\existsR$-complete, where $\exists\R$ is the well-studied class defined as the closure of the
Existential Theory of the Reals (ETR) under polynomial time many-one reductions.

In this work, we study the satisfiability (and validity) problem for such formulas involving probabilities. As already seen above, the choice of the admissible operations have an impact on the complexity. There are further choices, which we can make and which we will discuss in the following. All these choices have an impact on the complexity. Therefore, we will phrase these satisfiability problems as multi-parametric or multi-dimensional problems. The first dimension is the expressiveness of the underlying arithmetic:

\begin{description}
    \item[$\bullet$ Arithmetic:] Basic, linear, or polynomial. 
\end{description}

The second dimension will be the layer in the Pearl’s Causal Hierarchy which we have already discussed in the previous section. Reasoning about uncertainty has been the subject of intensive studies in the past decades. The research has resulted  in a large number of significant achievements, especially in the case of probabilistic inference with the input probability distributions encoded by Bayesian 
networks~\cite{pearl1988probabilistic}. The setting is of great interest in AI since it allows one not only to reason about \emph{observations}, but also about \emph{causality},  and even \emph{counterfactuals}, which is for instance important in fairness considerations. These three layers form a hierarchy of increasing expressiveness (for a more detailed description of the concepts, see Section~\ref{sec:preliminaries}).

%


\begin{description}
    \item[$\bullet$ PCH layer:] Probabilistic (observational), causal (interventional), or counterfactual.
\end{description}

The languages used in~\cite{fagin1990logic,ibeling2022mosse} are capable of fully expressing probabilistic reasoning, in particular, they allow to express \emph{marginalization}, which is a common paradigm in this field. However, in the frameworks discussed above, this ability is very
limited since the languages \emph{do not} allow the use of a unary~summation~operator $\Sigma$. Thus, for instance\footnote{
We have chosen this example because it is short. In this particular example, the joint probability could be directly written as $\PP{y}$ in the calculus of Fagin et al., see the formal definition in Section~\ref{sec:preliminaries}. This is however not true anymore if we sum over more complicated expressions.}, to express the marginal distribution of a random variable $Y$ over a subset of (binary) variables $\{Z_1,\ldots,Z_m\}$ as $\sum_{z_1,\ldots,z_m} \PP{y,z_1,\ldots,z_m}$, an encoding without summation requires an expansion into $\PP{y,Z_1=0,\ldots,Z_m=0} + \ldots +\PP{y,Z_1=1,\ldots,Z_m=1}$, which of exponential size in $m$. Consequently, to analyze the complexity aspects of the studied problems, languages that exploit the standard notation of probability theory and statistics, one requires an encoding that represents marginalization via the sum operator $\Sigma$. In the recent paper \cite{zander2023ijcai}, the authors present the first systematic study in this setting. They introduce a new natural class, named $\succR$, which can be viewed as a succinct variant of $\existsR$. 

%
%

Given the impact of the compact marginalization operator, the way we can perform marginalization will be the third dimension: 

\begin{description}
    \item[$\bullet$ Marginalization:] Expanded or compact. 
\end{description}

For expanded marginalization, the complexity results by Fagin et al. \cite{fagin1990logic}, as well as Moss{\'e} \cite{ibeling2022mosse} et al., can be easily adapted to capture all other parameter ranges, see \procOrFullVersion{the full version of our work}{Appendix~\ref{app:expanded}}. Therefore, throughout this paper, we only deal with the case when we have a compact marginalization operator in our calculus.

The fact that the decision problem for basic and linear arithmetic without summation is in $\NP$ is not obvious. A model for a formula is a joint probability distribution of $n$, say, binary variables, which is a table with $2^n$ entries. Furthermore, it is a priori not clear that the bit size of the entries is polynomial. It turns out, that these problems 
have what is called a \emph{small model property}. If there is a model, that is, a joint probability distribution satisfying the expression, then there is one which has only polynomially many non-zero entries. This follows by linear algebra arguments. If the model is small, then it can be guessed efficiently 
in $\NP$. In the case of polynomial arithmetic, we still have the small model property, but one cannot guarantee that the entries have polynomial size.
However, the compact marginalization operator destroys the small model property that was crucial in the work of \cite{fagin1990logic} and \cite{ibeling2022mosse} and in fact increases the complexity of the 
satisfiability problems dramatically \citep{zander2023ijcai}. So the next dimension is the model size, which can either be polynomially bounded or arbitrary:

\begin{description}
    \item[$\bullet$  Model size:] Small (polynomially many entries) or unbounded. 
\end{description}

The impact of restricting the model size in the presence of a compact marginalization operator will be one of the main topics in this work (see the results in Section~\ref{sec:small}). 

Causal models are often represented as a graph or Bayesian network where all random variables are represented as nodes and the edges encode which variables influence each other. E.g., at the probabilistic level, the graph would encode which variables are conditionally independent of each other.
In the satisfiability problems described so far, the graph was not mentioned, so the task was to decide whether there exists any model with any graph structure that satisfies the formula. 
In many studies, a researcher can infer information about the graph structure, combining observed and experimental data with existing knowledge. 
Indeed, 
learning graphical, causal structures from data has been the subject of a considerable amount of research (see, e.g., \cite{meinshausen2016methods,drton2017structure,glymour2019review,squires2022causal} for recent reviews).
 Thus, we will study the satisfiability (resp. validity) of formulas 
 with the additional requirement that the graph structure $\cG$ is also given in the input. 
The goal would be to determine whether the formula 
is satisfied in a model (resp. is valid in all models) 
having the structure~$\cG$.

In fact, such a constraint validity problem is one of the central problems of causality. E.g., given a graph structure, Pearl's  prominent do-calculus \cite{Pearl2009} is often applied to show the validity of the equivalence between causal and probabilistic expressions.
The  impact of including the graph in the input on the complexity will be the second important question that we will study in this work (see Section~\ref{sec:given:DAG:structure}). Thus, as a fifth dimension, we study the model structure:

\begin{description}
    \item[$\bullet$  Model structure:] Specified by a graph as a part of the input or unconstrained.
\end{description}

\subsection{A brief overview of our results}

We continue by giving an informal description of our results. A detailed description can be found in Section~\ref{sec:results}. 

We first show that constraining the graph structure of the model makes the problem only harder in the sense that the unconstrained problem can be reduced to the constrained one, since the unconstrained problem can always be reduced to the constrained one with the complete directed acyclic graph as a model. It turns out that for the probabilistic or interventional layer together with basic or linear arithmetic, there is indeed a jump. The probabilistic layer becomes at least $\existsR$-hard, while  the interventional layer jumps from $\PSPACE$ to $\NEXP$. For all other combinations of the PCH layer and arithmetic, no such jump happens and the complexity stays the same, however, we require new proof techniques. 
The exact results are given compactly in Table~\ref{fig:graph} (for comparison, the complexity landscape for unconstrained models, can be found in Table~\ref{fig:unconstrained}).

Second, we consider the case when the underlying model is small; small here means that for the underlying probability distribution $P$, the number of tuples $(u_1,\dots,u_m)$ with $P(U_1 = u_1,\dots,U_m = u_m) > 0$ is polynomially bounded in the input size. For the interventional and counterfactual layer with polynomial arithmetic, this does not affect the complexity of the problem (again, new proof techniques are required), but for the probabilistic layer, the complexity drops down to $\existsR^\Sigma$, a new class contained in $\PSPACE$.
The full results are given compactly in Table~\ref{fig:smallmodel}.

\section{Preliminaries}
\label{sec:preliminaries}

\subsection{Pearl’s Causal Hierarchy: An example}

\label{sec:example:PCH}

To illustrate the main ideas behind Pearl’s Causal Hierarchy (PCH), we present below
an example that, we hope, will make it easier to understand the formal definitions given in Section~\ref{sec:appendix:formal:definitions:syntax:and:semantics}.
In this example, we consider a hypothetical scenario involving three attributes represented by binary random variables:
Age, modeled by $Z=1$ (young) and $Z=0$ (old), (COVID-19) Vaccination represented by $X=1$, if yes, and $X=0$ if not vaccinated, and Recovery, represented by $Y=1$ (and $Y=0$ meaning mortality).
Below, we describe a structural causal model (SCM) that represents an unobserved true mechanism underlying this scenario and illustrates the canonical patterns of reasoning expressible at different levels of the PCH.

\begin{figure}[h] 
\vspace*{6mm}
\begin{center}
\footnotesize
$\begin{array}{c@{\hskip 1.2mm}c@{\hskip 1.2mm}c@{\hskip 1.2mm}|c|c@{\hskip 1.2mm}c@{\hskip 1.2mm}c@{\hskip 1.2mm}}\vspace*{-8mm}\\
   U_1&U_2&U_3&\multicolumn{1}{c|}{P(\bu)}&Z&X&Y\\ \hline \hline
 0& 0&0 &0.0474&0 &1&0\\
 0& 0&1 &0.4266&0 &1&1\\
 0& 1&0 &0.0126&0 &0&1\\
 0& 1&1 &0.1134&0 &0&0\\
 1& 0&0 &0.0316&1 &0&1\\
 1& 0&1 &0.2844&1 &0&0\\
 1& 1&0 &0.0084&1 &1&1\\
 1& 1&1 &0.0756&1 &1&1\\
 \multicolumn{7}{c}{(a)~ \text{Hidden}}
   \end{array}$
   \hspace*{2mm}
   $\begin{array}{c@{\hskip 1.2mm}c@{\hskip 1.2mm}c@{\hskip 1.2mm}|c}\vspace*{-8mm}\\
   Z& X&Y&\multicolumn{1}{c}{P(z,x,y)}\\ \hline \hline
0 &0&0& 0.1134\\
0 &0&1& 0.0126\\  
0 &1&0&  0.0474\\
0 &1&1& 0.4266\\
1 &0&0& 0.2844\\
1 &0&1& 0.0316\\
1 &1&1& 0.0840\\
\multicolumn{4}{c}{}\\
\multicolumn{4}{c}{(b)~\text{Observed $P$}}
   \end{array}
 \hspace*{3mm}
   \begin{array}{c|c@{\hskip 1.2mm}c@{\hskip 1.2mm}c@{\hskip 1.2mm}}\vspace*{-8mm}\\
   \multicolumn{1}{c|}{P(\bu)}&Z&X&Y\\ \hline \hline
 0.0474&0 &1&0\\
 0.4266&0 &1&1\\
 0.0126&0 &1&0\\
0.1134 &0 &1&1\\
0.0316 &1 &1&1\\
0.2844 &1 &1&1\\
 0.0084&1 &1&1\\
0.0756 &1 &1&1\\
\multicolumn{4}{c}{(c)~\text{\it do$(X=1)$}}
   \end{array}
   \hspace*{2mm} 
   \begin{array}{c@{\hskip 1.2mm}c@{\hskip 1.2mm}|c}\vspace*{-8mm}\\
   Z&Y&\multicolumn{1}{c}{P([X=1]z,y)}\\ \hline \hline
   0&0&   0.06 \\ 
   0&1&   0.54\\ 
   1&0&   0.00  \\
   1&1&   0.40\\ 
   \multicolumn{3}{c}{}\\
   \multicolumn{3}{c}{}\\
   \multicolumn{3}{c}{}\\
   \multicolumn{3}{c}{}\\
   \multicolumn{3}{c}{(d)~\text{Post-int. $P$}}
   \end{array}
   $
\end{center}
\caption{$(a)$ \emph{Unobserved true nature}: Probabilities of the unobserved variables and the induced by the mechanism $\cF$ outcomes of the observed variables. $(b)$ \emph{Observational layer}:  Probability distribution $P(z,x,y)$ of the observed variables. $(c)$  \emph{Interventional layer}: The modified mechanism $\cF$ by using the operator {\it do}$(X=1)$ (shown in $(c)$) leads to the \emph{post-interventional} distribution
$(d)$, denoted as $\pp{[X=1]z,y}$, for $z,y\in\{0,1\}$.
A challenging task here is to compute the post-interventional distribution $(d)$ from the observed distribution $(b)$.
\label{fig:ex:1}}
\end{figure}

\noindent{\bf Structural Causal Model.} An SCM is defined as a tuple $(\cF,P,\bU, \bX)$ which is of \emph{unobserved nature} from the perspective of an empirical researcher. 
It specifies 
the distribution $P(\bU)$ of the population and the mechanism $\cF$. In our example, the model assumes three independent binary random variables $\bU=\{U_1,U_2,U_3\}$, with probabilities:
$
\pp{U_1=1} = 0.4,
\pp{U_2=1} = 0.21,
\pp{U_3=1} = 0.9
$. 
They affect the observed endogenous (observed) random variables $Z,X,Y$ via the mechanism $\cF=\{F_1,F_2,F_3\}$ specified as follows: 
$Z:=F_1(U_1) = U_1;$\ 
$X:=F_2(Z,U_2) = Z U_2 + (1-Z)(1-U_2);$\ 
 $Y:=F_3(Z,X,U_3) =
XZ + (1-X)(1-U_3) +  X(1-Z)U_3$.

Thus, the model determines the distribution 
$P(\bu)$, for $\bu=(u_1,u_2,u_3)$, 
and the values for the observed variables, 
as can be seen in Fig.~\ref{fig:ex:1}$(a)$.

The unobserved random variable $U_1$ reflects the age of the population  and $Z$ is a function of $U_1$ (which may be more complex in a real population).
Getting COVID-19 vaccination depends on age but also on other circumstances (severe allergic reaction to a component of the COVID-19 vaccine, moderate or severe acute illness, etc.) and this is modeled by 
$U_2$. So $X$ is a function of $Z$ and $U_2$.
Finally, $Y$ depends on age, being vaccinated, and
on further circumstances like having other diseases, which are modeled by $U_3$. So $Y$ is a function of $Z$, $X$, and $U_3$. The functions $F_i$ define a directed graph structure on the random (observed) variables: there is an edge from $A$ to $B$ if $B$ depends on $A$. We always assume that
dependency graph of the SCM is acyclic. This property is also called \emph{semi-Markovian}.
In our example, we get the DAG:
\begin{tikzpicture}[baseline=-.3em] 
\node (s) at (0,0) {$Z$};
\node (t) at (1.1,0) {$X$};
\node (r) at (2.2,0) {$Y$.};
\draw [->] (t) -- (r);
\draw [->] (s) -- (t);
\draw [->] (s) edge [bend left=27] (r);
\end{tikzpicture}\\

%
\noindent{\bf Layer 1 (probabilistic).} Empirical sciences primarily rely on observed data, typically represented as probability distributions over measurable variables. In our example, this corresponds to the distribution 
$P(Z,X,Y)$. The unobserved variables $U_1,U_2,U_3$ along with the causal mechanism $\cF$, remain hidden from direct observation. A researcher thus receives the probabilities probabilities (shown in Fig.~\ref{fig:ex:1}$(b)$)
 $P(z,x,y)=\sum_{\bu}\delta_{\cF,\bu}(z,x,y)\cdot P(\bu)$, where vectors
 $\bu=(u_1,u_2,u_3)\in \{0,1\}^3$~and
 $\delta_{\cF,\bu}(z,x,y)=1$~if $\compactEquals{F_1(u_1)=z, F_2(z,u_2)=x,}$ 
 and 
 $\compactEquals{F_3(z,x,u_3)=y}$; otherwise
 $\delta_{\cF,\bu}(z,x,y)=0.$
The relevant query in our scenario $\pp{Y=1\mmid X=1}$ can be evaluated as
$\pp{Y=1\mmid X=1}=\pp{Y=1,X=1}/\pp{X=1}=0.5106/0.558\approx 0.915$  which says 
that the probability for recovery ($\compactEquals{Y=1}$) is high
given that the patient was vaccinated ($\compactEquals{X=1}$).
On the other hand, the query for $\compactEquals{X=0}$ can be evaluated as
$\pp{Y=1\mmid X=0}=0.0442/0.442=0.1$, which
may lead to the opinion that the vaccine is relevant to recovery. However, these results do not reflect a randomized controlled trial \cite{fisher1936design} that could establish the differences between a COVID-19 vaccination and a placebo vaccination.\\

\noindent{\bf Layer 2 (interventional).} Consider a randomized  trial in which each patient is vaccinated,
denoted as $\compactEquals{\dop(X=1)}$, regardless of age ($Z$) and other conditions ($U_2$). 
We model this by performing a hypothetical intervention in which we replace in $\cF$ 
the mechanism $F_2(Z,U_2)$ by the constant function $1$ and leaving the remaining functions unchanged (see, Fig.~\ref{fig:ex:1}$(c)$).
If $\cF_{X=1} \ =\{\compactEquals{F'_1=F_1,F'_2=1,F'_3=F_3}\}$ 
denotes the new mechanism, then the \emph{post-interventional} 
distribution
$\pp{[X=1]Z,Y}$ is specified as 
$\pp{[X=1]z,y}=\sum_{\bu}\delta_{\cF_{X=1},\bu}(z,y)\cdot P(\bu),$ where 
$\delta_{\cF_{X=1},\bu}$ denotes function $\delta$ as 
above, but for the new 
mechanism $\cF_{X=1}$ (the distribution is shown in Fig.~\ref{fig:ex:1}$(d)$)\footnote{A common and popular notation for the post-interventional probability 
is $\pp{Z,Y\mmid \compactEquals{\dop(X=1)}}$. In this paper,
we use the notation $\pp{[\compactEquals{X=1}]Z,Y}$ since
it is more convenient for analyses involving counterfactuals.}.
To determine the causal effect of the COVID-19 vaccination  on recovery, we 
compute, in an analogous way, the distribution $\pp{[X=0]Z,Y}$ 
after  the intervention $\compactEquals{\dop(X=0)}$, which means that all patients receive a placebo vaccination. 
Then, comparing the value 
$\pp{[X=1]Y=1}=0.95$ with $\pp{[X=0]Y=1}=0.0874,$
we can conclude that $\pp{[X=1]Y=1} - \pp{[X=0]Y=1} >0$. 
This can be
interpreted as a positive  (average) effect of the vaccine in the population. 
Note that it is not obvious how to compute 
the~post-interventional distributions 
from the observed probability $P(Z,X,Y)$; Indeed, this is 
a challenging task in the field of causality. \\

\noindent{\bf Layer 3 (counterfactual).} The key phenomena that can be modeled and analyzed at this level 
are counterfactual situations. Imagine, e.g., in our scenario, there is a group of patients who have not been vaccinated
and died $\compactEquals{(X = 0, Y = 0)}$.
One may ask, what the outcome would be  $Y$ had they been vaccinated
$\compactEquals{(X = 1)}$. In particular,  one can ask what the probability of recovery would be 
if we had vaccinated patients in this group.
Using the formalism of  Layer~3, we can express this as 
a counterfactual query:
$\pp{[X=1]Y=1 \mmid X=0,Y=0}= \pp{[X=1](Y=1) \wedge (X=0,Y=0)} / \pp{X=0, Y=0}.$
Note that  the event $\compactEquals{[X=1](Y=1) \wedge (X=0,Y=0)}$ incorporates 
simultaneously two counterfactual mechanisms: $\cF_{X=1}$ and~$\cF$. 
This is the key difference to Layer 2, where we can only have one.

\subsection{Syntax and semantics of probabilistic calculi}
\label{sec:appendix:formal:definitions:syntax:and:semantics}

In this section, we describe the underlying formalism of the satisfiability problems that we consider. This will cover the first three parameter dimensions: the underlying arithmetic, with or without compact marginalization, and the PCH level. The other two dimensions then naturally follow since they are obtained by restricting the underlying model.

We always consider discrete distributions 
in the probabilistic and causal languages studied in this paper. We
represent the values  of the random variables as $\mathit{Val} = \{0,1,\myldots, \maxvaluecount - 1\}$ 
and denote by $\bX$ the set of random variables used in a system. 
Here, we use bold letters for sets of variables or sets of values.
By capital letters $X_1,X_2, \myldots$, we denote the individual variables 
and assume, w.l.o.g.,~that they all share the same domain  $\mathit{Val}$.
A value of $X_i$ is often denoted by the corresponding lowercase letter $x_i$ or a natural number.
In this section, we describe the syntax and semantics of the languages starting 
with probabilistic ones and then we provide extensions to the causal systems.

By an \emph{atomic} event, we mean an event of the form $X=x$, where 
$X$ is a random variable and $x$ is a value in the domain of $X$. 
The language $\Eprop$ of propositional formulas over atomic events is defined 
as the closure of such events under the Boolean operators $\wedge$ and $\neg$. 
To specify the syntax of interventional and counterfactual events, we 
define the intervention and extend the syntax of $\Eprop$ to $\Epint$ and $\Ecounter$,
respectively,
using the following grammars:
\[
\begin{array}{rccl}
 \mbox{$\Eprop$ is defined by} 	& \bp 	& ::= & X = x  \mid \neg \bp \mid \bp \wedge \bp
 	 \\[1mm]
\mbox{$\Eint$ is defined by}    	& \bi		& ::= & \top \mid  X = x       \mid  \bi \wedge \bi 
	 \\[1mm]
\mbox{$\Epint$ is defined by} 	&\bp_{\bi} 	& ::= & [\, \bi\, ]\, \bp \\[1mm]
  \mbox{$\Ecounter$ is defined by}	& \bc 	& ::= & \bp_{\bi}  \mid \neg \bc \mid \bc \wedge \bc.
\end{array}
\]
Note that since $\top$ means that no intervention has been applied,
we can assume that $\Eprop \subseteq \Epint$.

The PCH 
consists of three languages 
$\cL_1, \cL_2,\cL_3$, each of which is based on terms of the form $\PP{\delta}$.
For the (observational or associational) language $\cL_1$, we have  $\delta\in \Eprop$,
for the (interventional) language $\cL_2$,  we have $\delta\in \Epint$ and 
for the (counterfactual) language $\cL_3$, $\delta\in \Ecounter$. 
The expressive power and computational complexity properties of the languages depend 
largely on the operations that we are allowed to apply to the basic terms.
Allowing gradually more complex operators, we describe 
the languages which are the subject of our studies below. We start with the description of
the languages $\cT_i^*$ of terms, with $i=1,2,3$, using the following grammars\footnote{In the given grammars, we omit the brackets for readability, but we assume that they can be used in a standard way.}
\[
\begin{array}{ll c|c ll}

    \multirow{1.5}{*}{$\Ticomp$}  & \multirow{1.5}{*}{$\bt::=\PP{\delta_i}$}   &\multirow{1.5}{*}{}&\multirow{1.5}{*}{}&  
        \multirow{1.5}{*}{$\Ticompsum$}  & \multirow{1.5}{*}{$\bt::=\PP{\delta_i}   \mid  \mbox{$\sum_{x} \bt$} $} \\[1.6mm]
    \Tilin  & \bt::=\PP{\delta_i}    \mid \bt +\bt       &&&
    	 \Tilinsum  & \bt::=\PP{\delta_i}   \mid \bt +\bt  \mid  \mbox{$\sum_{x} \bt$} \\[1mm]
   \Tipoly & \bt::=\PP{\delta_i} \mid \bt+\bt  \mid -\bt \mid \bt \cdot \bt &&&
   \Tipolysum &\bt::= \PP{\delta_i} \mid \bt + \bt  \mid -\bt \mid \bt \cdot \bt \mid 
 \mbox{$\sum_{x} \bt$} 
 
   \end{array}
\]
where $\delta_1$ are formulas in $\Eprop$, $\delta_2\in\Epint$, 
$\delta_3\in \Ecounter$.

The probabilities of the form $\PP{\delta_i}$ 
are called \emph{primitives} or \emph{basic terms}.
In the summation operator $\sum_{x}$, we have 
a dummy variable $x$ 
which ranges over all values $0,1,\ldots, \maxvaluecount - 1$.
The summation $\sum_{x} \bt$ is a purely syntactical 
concept which represents the sum 
$\bt[\sfrac{0}{ x}]  +\bt[\sfrac{1}{x}]+\myldots +\bt[\sfrac{\maxvaluecount - 1}{x}]$,
whereby $\bt[\sfrac{v}{x}]$, we mean the expression in which all occurrences of $x$
are replaced with value $v$.
For example,  for  $\mathit{Val} = \{0,1\}$,
the expression
 $\sum_{x} \PP{Y=1, X=x}$
 semantically represents $\PP{Y=1, X=0} + \PP{Y=1, X=1}$.

We note that the dummy variable $x$ is not a (random) variable in the usual sense
and that its scope is defined in the standard way.

In the table above, the terms in $\Ticomp$ are just basic probabilities with the events given by the corresponding languages $\Eprop$, $\Epint$, or $\Ecounter$. 
Next, we extend terms by being able to compute sums of probabilities and by 
adding the same term several times, we also allow for weighted sums with 
weights given in unary. Note that this is enough to state all our hardness
results. All matching upper bounds 
also work when we allow for explicit weights given in binary.
In the case of $\Tipoly$, we are allowed to build polynomial terms in the 
primitives. On the right-hand side of the table, we have the same
three kinds of terms, but to each of them, we add a marginalization operator
as a building block.

The polynomial calculus $\Tipoly$ was originally introduced by Fagin, Halpern, and Megiddo \citep{fagin1990logic}
(for $i = 1$) 
to be able to express conditional probabilities by clearing denominators. While this works 
for $\Tipoly$, this does not work in the case of $\Tipolysum$,
since clearing denominators with exponential sums creates expressions that
are too large. But we could introduce basic terms 
of the form $\PP{\delta_i \mmid \delta}$ with $\delta \in \Eprop$
explicitly. All our hardness proofs work without conditional probabilities
but all our matching upper bounds are still true with explicit
conditional probabilities.

Expressions like
$\PP{X=1} + \PP{Y=2} \cdot \PP{Y=3}$
are valid terms in $\Tprobpoly$ and 
$\sum_z\PP{[X=0](Y=1, Z=z)}$ and $\sum_z\PP{([X=0]Y=1), Z=z}$ are
valid terms in the language  
$\Tcausalpolysum$, for example.

Now, let 
$
\Lab= \{
\text{base}, \text{base}\langle{\Sigma}\rangle, 
\text{lin}, \text{lin}\langle{\Sigma}\rangle,     
\text{poly}, \text{poly}\langle{\Sigma}\rangle
\}
$
denote the labels of all variants of languages. Then for each   $*\in\Lab$ and $i=1,2,3$, we define
the languages $\Listar$ of Boolean combinations of inequalities in a standard way: 
\begin{align*}
   & \mbox{$\Listar$ is defined by}\quad  \bff ::= \bt \le \bt' \mid \neg \bff \mid \bff \wedge \bff , \quad
    \mbox{where $ \bt,\bt'$ are terms in  $ {\cal T}^{*}_{i} $.}
\end{align*}

Although the language and its operations can appear rather restricted, all the usual elements of probabilistic and causal formulas can be encoded. Namely, equality is encoded as greater-or-equal in both directions, e.g. 
$\PP{x} = \PP{y}$ means $\PP{x} \geq \PP{y} \wedge \PP{y} \geq \PP{x}$.

The number~$0$ can be encoded as an inconsistent probability, 
i.e., $\PP{X=1 \wedge X=2}$. 
In a language allowing addition and multiplication, any positive integer can be easily encoded
from the fact $\PP{\top} \equiv 1$, e.g. $4 \equiv (1 + 1) (1 + 1) \equiv (\PP{\top} + \PP{\top}) (\PP{\top} + \PP{\top})$.
If a language does not allow multiplication, one can show that the encoding is still possible.
Note that these encodings barely change the size of the expressions, so allowing or disallowing these additional operators does not affect any complexity results involving these expressions.

To define the semantics of the languages,  we use  a structural causal model (SCM) as in \cite[Sec.~3.2]{Pearl2009}.
An SCM 
is a tuple $\fM=(\cF, P, $ $\bU, \bX)$, such that $\bV = \bU \cup \bX$ is a set of 
variables partitioned into 
exogenous (unobserved) variables $\bU=\{U_1,U_2,\myldots \}$ and 
endogenous variables $\bX$.
The tuple $\cF=\{F_1,\myldots,F_n\}$ consists of
functions such that function $F_i$ calculates the value of variable $X_i$ from the values 
$(\bx, \bu)$ of other variables in $\bV$ as  
$F_i(\pa_i,\bu_i)$ \footnote{We consider recursive models, 
that is, we assume the endogenous variables 
are ordered such that variable $X_i$ (i.e. function $F_i$) is not affected by any 
 $X_j$ with $j > i$. 
Here, we also use the usual notation with capital letters and lowercase letters where  $\Pa_i$  are variables and $\pa_i$ is an assignment of values to those variables.},
 where $\Pa_i\subseteq \bX$ and $\bU_i\subseteq \bU$. 
$P$ specifies a probability distribution 
of all exogenous  variables $\bU$. Since variables $\bX$  depend deterministically on 
 the exogenous variables via functions $F_i$,
 $\cF$ and $P$ obviously define 
the joint probability distribution of 
$\bX$.
 
 Throughout this paper, we assume that domains of endogenous variables $\bX$ are 
 discrete and finite. In this setting, exogenous variables $\bU$ could take values 
 in any  domain, including infinite and continuous ones.
 A recent paper \citep{zhang2022partial} shows, however,
 that any SCM over discrete endogenous variables is equivalent 
 for evaluating post-interventional probabilities to an SCM where all exogenous variables are discrete with finite
 domains.

As a consequence, throughout this paper, we assume that domains of 
exogenous variables $\bU$ are discrete and finite, too.

For any basic  $\Eint$-formula $\let\IsInPP=1 X_i=x_i$ 
(which, in our notation, means $\let\IsInPP=1 \dop(X_i=x_i)$),  we denote 
by $\cF_{X_i=x_i}$ the function obtained from $\cF$ by replacing $F_i$
with the constant function $F_i(\bv):=x_i$. 
We generalize this definition for any interventions specified by
$\alpha\in \Eint$ in a natural way and denote as 
$\cF_{\alpha}$ the resulting functions.
For any $\varphi\in \Eprop$,  we write $\cF, \bu \models \varphi$
if $\varphi$ is satisfied for values of $\bX$ calculated from the values $\bu$.
For $\alpha\in \Eint$, we write $\cF, \bu \models [\alpha]\varphi$ if  
$\cF_{\alpha}, \bu \models \varphi$. And for all $\psi,\psi_1,\psi_2\in \Ecounter$,
we write  $(i)$ $\cF, \bu \models \neg\psi$ if  $\cF, \bu \not\models \psi$ and 
$(ii)$ $\cF, \bu \models \psi_1 \wedge \psi_2$ if  $\cF, \bu \models \psi_1$
and $\cF, \bu \models \psi_2$.
  
Finally, for $\psi\in  \Ecounter$, let $S_{\fM}(\psi)=\{\bu \mid \cF, \bu \models \psi\}$.

We define $\llbracket \be \rrbracket_{\fM}$, for some expression $\be$,
recursively in a natural way, 
starting with basic terms as follows 
$\llbracket \PP{\psi} \rrbracket_{\fM} = \sum_{\bu\in S_{\fM}(\psi)}P(\bu)$
and, for $\delta\in\Eprop$, $\llbracket \PP{\psi\mmid \delta} \rrbracket_{\fM} = \llbracket\PP{\psi \wedge \delta} \rrbracket_{\fM}/ \llbracket\PP{\delta} \rrbracket_{\fM}$, assuming that the expression is undefined if  $\llbracket\PP{\delta} \rrbracket_{\fM}=0$.
For two expressions $\be_1$ and $\be_2$, we define 
$ \fM \models  \be_1 \le  \be_2$,  if and only if, 
$\llbracket \be_1 \rrbracket_{\fM}\le \llbracket \be_2 \rrbracket_{\fM}.$
The semantics for negation and conjunction are defined in the usual way,
giving the semantics for $\fM \models \varphi$ for any formula $\varphi$
in $\Lcausalstar$.

\subsection{Probabilistic and causal satisfiability problems}

The (decision) satisfiability problems for languages of the PCH, denoted by $\SATistar$,  
with $i=1,2,3$ and $*\in \Lab$,
take as input a formula $\varphi$ in 
$\Listar$ and  ask whether there exists a model 
$\fM$ such that $\fM \models\varphi$.
Analogously, the validity problems for $\Listar$ consists in deciding whether, for a given $\varphi$,
$\fM \models\varphi$ holds
for all models  $\fM$.
From the definitions, it is obvious that variants of the problems for the level $i$
are at least as hard as the  counterparts at the lower level.

In many studies, a researcher can infer information about the graph structure of the 
SCM, combining observed and experimental data with existing knowledge. 
Indeed, learning causal structures from data has been the subject of a considerable 
amount of research (see, e.g., \cite{drton2017structure,glymour2019review,squires2022causal}
for recent reviews and \cite{reisach2021beware,ng2021reliable,lu2021improving,rolland2022score,reisach2023scale} for current achievements). Thus, the problem of interest in this setting, which is a subject of this section, 
is to determine the satisfiability, resp., validity of formulas in SCMs with the additional 
requirement on the graph structure of the models.

Let $\fM=(\cF=\{F_1,\ldots,F_n\}, P, \bU, \bX=\{X_1,\ldots,X_n\})$ be an SCM.
We will assume that the models are \emph{semi-Markovian} in the general case and \emph{Markovian} in the graph constrained case\footnote{We note that 
the general semi-Markovian model,  which allows 
for the sharing of exogenous arguments and allows for arbitrary dependencies among 
the exogenous variables, can be reduced in a standard way to the Markovian model by 
introducing new auxiliary ``latent'' variables $L_1,\ldots,L_k$ which belong to the endogenous variables $\bX$ of the model, but which are non-measurable in contrast to the ``standard'' observed / measurable variables. Then we can use such latent variables as nodes of a DAG but the joint probability distribution on $\bX$ is restricted  to only the measurable variables and ignores $L_1,\ldots,L_k$.
}, 
Markovian means that the exogenous arguments $U_i, U_j$ of $F_i$, resp.~$F_j$ are independent whenever 
$i \not= j$. We define that a DAG $\cG = (\bX, E)$ represents the graph structure of $\fM$ if, 
for every $X_j$ appearing as an argument of $F_i$, $X_j \to X_i$ is an edge in~$E$.
DAG $\cG$ is called a \emph{causal diagram} of the model $\fM$ \cite{Pearl2009,bareinboim2022pearl}.
The satisfiability problems with an additional requirement on the causal diagram, 
take as input a formula $\varphi$ and a DAG $\cG$ with the task of deciding whether there exists an SCM
with the structure $\cG$ for $\varphi$. In this way, for problems $\SATistar$, we 
get the corresponding problems denoted as $\SATistargraph$, with  $*\in \Lab$.

An important feature of $\SATipoly$ instances, over all levels $i=1,2,3$, is the small 
model property which says that, for every satisfiable formula, there is a model whose size is 
bounded polynomially with respect to the length of the input. This property 
was used
to prove the memberships of $\SATipoly$ in $\NP$ and in $\existsR$
\citep{fagin1990logic,ibeling2020probabilistic,ibeling2022mosse,IBELING2024103339}.
Interestingly, membership proofs of $\SATproblinsum${} in \NP$^{\ccPP}$ and  $\SATinterventlinsum$ in $\PSPACE$  \citep{doerflerICLR2025} 
rely on the property~as~well.

Apart from these advantages, the small model property is interesting in itself.
For example, on the probabilistic layer, if a formula $\varphi$ (without the summation)
over variables $X_1,\ldots,X_n$ is satisfiable, then there exists a Bayesian network 
$\cB=(\cG,P_{\cB})$, with a DAG $\cG$ over nodes $X_1,\ldots,X_n$ and 
conditional probability tables $P_{\cB}$ for each variable $X_i$, such that $\varphi$ is true 
in $\cB$ and the total size of the tables is bounded by a polynomial of  $|\varphi|$.

These have motivated the introduction of the small-model problem, denoted as 
$\SATprobpolysumsm$, which is defined like $\SATprobpolysum$ with the additional constraint 
that a satisfying distribution should only have polynomially large support, that is, only polynomially 
many entries in the exponentially large table of probabilities are non-zero \citep{blaser2024existential}.
In the paper, the authors achieve this by extending an instance with an additional unary input $p \in \IN$
and requiring that the satisfying distribution has a support of size at most $p$.
We define $\SATinterventpolysumsm$ and $\SATcausalpolysumsm$ in a similar way.
Formally, we use the following:
\begin{definition}\label{lab:def:small:model:new:U}
	The decision problems $\SATipolysumsm$, with $i=1,2,3$, take as input a formula 
	$\varphi \in \Lipolysum$  
	and a unary encoded number $p \in \IN$
	and ask whether there exists a model  $\fM=(\cF, P, \bU=\{U_1,\ldots,U_m\}, \bX)$ 
	such that $\fM \models\varphi$ and
 	$\ \#\{(u_1,\ldots,u_m): P(U_1=u_1,\ldots, U_m=u_m) > 0\}\le p.$
\end{definition}

\subsection{The (succinct) existential theory of the reals}

For two computational problems $A,B$, we will write $A\leqp B$ if $A$ can be reduced to $B$ in polynomial time, which means $A$ is not harder to solve than $B$. A problem $A$ is complete for a complexity class $\cC$, if $A \in \cC$ and, for every other problem $B\in\cC$, it holds $B\leqp A$. By ${\tt co}\mbox{-}\cC$, we denote the class of all problems
$A$ such that its complements $\overline{A}$ belong to $\cC$.

To measure the computational complexity of $\SATistar$, a central role play the
following, well-known Boolean complexity classes $\NP, \PSPACE, \NEXP,$ and $\EXPSPACE$   
(for formal definitions see, e.g., \cite{arora2009computational}).
Recent research has shown that the precise complexity of several natural 
satisfiability problems can be expressed in terms of the classes over the real numbers 
$\existsR$ and $\succR$.  Recall, that the existential theory of the reals $(\ETR)$ is the set of true sentences
of the form
\begin{equation} \label{eq:etr:1}
   \exists x_1 \dots \exists x_n \varphi(x_1,\dots,x_n),
\end{equation}
where $\varphi$ is a quantifier-free Boolean formula over the basis $\{\vee, \wedge, \neg\}$
and a signature consisting of the constants $0$ and $1$, the functional symbols
$+$ and $\cdot$, and the relational symbols $<$, $\le$, and $=$. The sentence
is interpreted over the real numbers in the standard way.
The theory forms its own complexity class $\exists \mathbb{R}$ which is 
defined as the closure of the ETR under polynomial time many-one reductions
\citep{grigoriev1988solving,existentialTheoryOfRealsCanny1988some,existentialTheoryOfRealsSchaefer2009complexity}.
A succinct variant of $\ETR$, denoted as  $\succETR$, and the corresponding 
class $\succR$, have been introduced in \cite{zander2023ijcai}.
$\succETR$ is the set of all Boolean circuits $C$ that encode 
a true sentence as in $(\ref{eq:etr:1})$ as follows. Assume that 
$C$ computes a function $\{0,1\}^N \to \{0,1\}^M$. Then $\{0,1\}^N$ 
represents the node 
set of the tree underlying $\varphi$
and $C(i)$ is an encoding of the description of node $i$, consisting of the label
of $i$, its parent, and its two children. The variables in 
$\varphi$ 
are $x_1,\dots,x_{2^N}$. As in the case of $\existsR$, to 
$\succR$ belong all languages 
which are polynomial time many-one reducible to $\succETR$. $\existsR$ and $\succR$ can also be defined in terms of real RAMs \cite{erickson2022smoothing, blaser2024existential}. These two classes correspond to nondeterministic polynomial and exponential time, respectively, on  such machines.

In a breakthrough result, Renegar \cite{renegar1992computational} showed that $\ETR \in \PSPACE$.
In addition, his result also shows that even if the formula contains an exponential number of arithmetic terms of exponential size with exponential degree, it is still possible to solve the formula in $\PSPACE$, as long as the number of variables is polynomially bounded.

\section{Results}
\label{sec:results}

Intuitively, one would expect that when we make the arithmetic of the underlying satisfiability problem more powerful, then the complexity of the problem will also increase. The same is true when we go up in the PCH. While this is sometimes the case, we will also see many interesting cases where the complexity remains unchanged when we change a certain parameter. Our main focus will be on the effect of constraining the model on the complexity of the problem. As mentioned earlier, we will always assume a compact marginalization operator since this is a very natural operation that frequently appears in expressions involving probabilities. Furthermore, without compact marginalization, everything follows from the results \cite{fagin1990logic, ibeling2022mosse} since the models appearing in their proofs are already very constrained (see \procOrFullVersion{the full version}{Appendix~\ref{app:expanded}}). 


\subsection{Previous work}

\begin{wraptable}{r}{0.58\textwidth}
\small
\begin{tabular}{|c    | c@{\hskip 0.05cm} | c@{\hskip 0.05cm} |c@{\hskip 0.0cm}   |}
\hline
 \multicolumn{1}{|c|}{\multirow{2}{*}{ Terms}} & 
  \multirow{2}{*}{$\cL_1$~(prob.)} &  \multirow{2}{*}{$\cL_2$~(interv.)} &  \multirow{2}{*}{$\cL_3$~(count.)  }   \\ 
  &&&
  \\
  \hline \hline 
&\multicolumn{3}{c|}{}\vspace*{-3.5mm}\\
\multicolumn{1}{|c|}{\multirow{1.2}{*}{basic}}
		&  \multicolumn{3}{c|}{\multirow{2.4}{*}{$~\NP$ ~~$(a)$}} 
		\\  \cline{1-1} 
\multicolumn{1}{|c|}{\multirow{1.2}{*}{lin} }
&  \multicolumn{3}{c|}{\multirow{2}{*}{}} 
				\\ \cline{1-1} \cline{2-4} 
\multicolumn{1}{|c|}{\multirow{2}{*}{poly}}
		&  \multicolumn{3}{c|}{\multirow{2}{*}{$
        \exists\R$~~~$(b)$}}
			\\
		 &\multicolumn{3}{c|}{}  \\  
     \hline \hline 
\multicolumn{1}{|c|}{\multirow{1.2}{*}{basic \& marg.}} & \multirow{2.4}{*}{$\NP^{\ccPP}$ $(c)$} & \multirow{2.4}{*}{$\PSPACE$ $(c)$} & 		
		\multirow{2.4}{*}{$\NEXP$ $(c)$}  
		\\ \cline{1-1} 
\multicolumn{1}{|c|}{ \multirow{1.2}{*}{lin \& marg.} }
& \multirow{2.4}{*}{} & \multirow{2.4}{*}{} & \multirow{2.4}{*}{}
 		\\  \cline{1-1}\cline{1-4}
\multicolumn{1}{|c|}{ \multirow{2}{*}{poly \& marg.} }
		& \multicolumn{3}{c|}{\multirow{2}{*}{ $
        \succR$~~~$(d)$}}\\
		\multicolumn{1}{|c|}{}&\multicolumn{3}{c|}{}  \\  \hline		
\end{tabular}\\[2mm]
\caption{The complexity landscape for \emph{unconstrained models}.
Sources: $(a)$ for $\cL_1$ \citep{fagin1990logic}, for $\cL_2$ and $\cL_3$ \citep{ibeling2022mosse},
$(b)$ \cite{ibeling2022mosse},
$(c)$ \cite{doerflerICLR2025},
$(d)$ for $\cL_1$ and $\cL_2$ \citep{zander2023ijcai}, for $\cL_3$ \cite{ibeling2024probabilistic,doerflerICLR2025}.
\label{fig:unconstrained}
}
\end{wraptable}

For unconstrained models, we have a clear picture of the satisfiability landscape, see Table~\ref{fig:unconstrained}. Fagin et al. \cite{fagin1990logic} prove the $\NP$-completeness for the probabilistic layer with basic and linear terms and the $\NP$-hardness with polynomial terms. Moss{\'e} et al.\ \cite{ibeling2022mosse} prove that the latter is in fact complete for the existential theory of the reals $\existsR$. They also prove that the same complexity results also hold for the interventional and counterfactual layers of the PCH. This means that the layer of the PCH does not have any impact on the complexity of the corresponding satisfiability problem. Things however change, when we add a marginalization operator. As shown in \cite{doerflerICLR2025} if the arithmetic is basic or linear, then the satisfiability problem is complete for $\NP^{\ccPP}$, $\PSPACE$, and $\NEXP$ when the PCH layer increases from probabilistic to counterfactual. If, however, the arithmetic is polynomial, then all three layers have the same complexity again, they are complete for the new class $\succR$ \citep{zander2023ijcai,ibeling2024probabilistic,doerflerICLR2025}.

\subsection{Our results}

As our main result, we complete the complexity landscape of the probabilistic and causal satisfiability problems when we either constrain the graph structure of the SCM or when we force the model to be small.

\begin{wraptable}{r}{0.68\textwidth}
\small
\begin{tabular}{|c    | c@{\hskip 0.05cm} | c@{\hskip 0.05cm} |c@{\hskip 0.0cm}   |}
\hline
 \multicolumn{1}{|c|}{\multirow{2}{*}{ Terms}} & 
  \multirow{2}{*}{$\cL_1$~(prob.)} &  \multirow{2}{*}{$\cL_2$~(interv.)} &  \multirow{2}{*}{$\cL_3$~(count.)  }   \\ 
  &&&
  \\
  \hline \hline 
\multicolumn{1}{|c|}{\multirow{1.2}{*}{basic \& marg.}} & \multirow{2.4}{*}{$(\NP^{\ccPP} \cup \existsR)$-hard $(a)$} & \multicolumn{2}{c|}{\multirow{2.4}{*}{$\NEXP$ $(b)$}} 
		\\ \cline{1-1} 
\multicolumn{1}{|c|}{ \multirow{1.2}{*}{lin \& marg.} }
& \multirow{2.4}{*}{} & \multicolumn{2}{c|}{\multirow{2.4}{*}{}} 
 		\\  \cline{1-1}\cline{1-4}
\multicolumn{1}{|c|}{ \multirow{2}{*}{poly \& marg.} }
		& \multicolumn{3}{c|}{\multirow{2}{*}{ $
        \succR$~~~$(c)$}}\\
		\multicolumn{1}{|c|}{}&\multicolumn{3}{c|}{}  \\  \hline		
\end{tabular} \\[2mm]
\caption{The complexity landscape for \emph{constrained graph structures}.
Sources: $(a)$ Lemma~\ref{prop:prob:comp:sum:graph:er:hard}, Proposition~\ref{prop:DAG:notharder}, and \cite{doerflerICLR2025}, $(b)$ Theorem~\ref{thm:SATinterventcompsumgraph:ccNEXP-complete}
$(c)$ Proposition~\ref{prop:DAG:notharder} and Theorem~\ref{thm:causal_with_graph:reduction}. 
\label{fig:graph}
}
\end{wraptable}
\noindent As the first part of our main results, we classify the difficulty of the satisfiability problems when the graph structure is given as a part of the input. We will show that these cases are always at least as hard as the corresponding problems without specified graph structures. This is due to the fact that we can always take the complete acyclic graph. 
The most notable difference to the unconstrained problems happens at the interventional layer. Here we have an increase in complexity from $\PSPACE$ to $\NEXP$. 
We also compare our results to the model checking problem. Here we are not only given a graph structure but the full SCM and we need to check whether the model satisfies the input formula. It turns out that model checking is considerably easier. For instance, when the arithmetic is polynomial with marginalization, then even at the probabilistic layer, the satisfiability problem is $\succR$-complete, while the corresponding model checking problem is in $\NP^\ccPP$ (Theorem \ref{thm:SATproblinbBN:DAG}).

\begin{wraptable}{r}{0.58\textwidth}
\small
\begin{tabular}{|c    | c@{\hskip 0.05cm} | c@{\hskip 0.05cm} |c@{\hskip 0.0cm}   |}
\hline
 \multicolumn{1}{|c|}{\multirow{2}{*}{ Terms}} & 
  \multirow{2}{*}{$\cL_1$~(prob.)} &  \multirow{2}{*}{$\cL_2$~(interv.)} &  \multirow{2}{*}{$\cL_3$~(count.)  }   \\ 
  &&&
  \\
  \hline \hline 
\multicolumn{1}{|c|}{
  \multirow{1.2}{*}{basic \& marg.}} & \multirow{2.4}{*}{~$\NP^\ccPP$~$(a)$} 
  & \multirow{2.4}{*}{$\PSPACE$~$(a)$} 
  & \multirow{2.4}{*}{$\NEXP$~$(a)$}  
		\\ \cline{1-1} 
\multicolumn{1}{|c|}{ \multirow{1.2}{*}{lin \& marg.} }
& \multirow{2.4}{*}{} & \multirow{2.4}{*}{} & \multirow{2.4}{*}{}
 		\\  \cline{1-1}\cline{1-4}
\multicolumn{1}{|c|}{ \multirow{2}{*}{poly \& marg.} }
		& \multicolumn{1}{c|}{\multirow{2}{*}{$\exists\R^{\Sigma}$~$(b)$}} & \multicolumn{2}{c|}{\multirow{2}{*}{ \NEXP~$(c)$}}\\
		\multicolumn{1}{|c|}{}& \multicolumn{1}{c|}{} &\multicolumn{2}{c|}{}  \\  \hline		
\end{tabular}\\[2mm]
\caption{The complexity landscape for \emph{satisfiability with small models}.  
Source: $(a)$ Lemma~\ref{lemma:small:model:in:linear} and \citep{doerflerICLR2025}, $(b)$ Fact~\ref{fact:small:model:sat:observed:unobserved} and \citep{blaser2024existential}, and $(c)$ Theorem~\ref{lab:lemma:interv:counter:smconj}
\label{fig:smallmodel}.
}
\end{wraptable}
Secondly, we consider the satisfiability problems with the small model property and investigate how the complexity increases across the PCH, in the presence of summation operators.  
In Definition~\ref{lab:def:small:model:new:U}, we defined the small model property for SCMs. 
In the original work of Fagin et al., there was no underlying SCM since they only dealt with the probabilistic layer, and a small model meant that we only have polynomially many entries in the table of the joint probability distribution of the (observed) variables. Since each observed variable is a function in the unobserved variables, both definitions are equivalent for semi-Markovian models as we show in Fact~\ref{fact:small:model:sat:observed:unobserved}. 
Most notably, with polynomial arithmetic, at the probabilistic layer, having a small model reduces the complexity to $\exists\R^{\Sigma}$, the existential theory of the reals enhanced with a compact summation operator. 
This was shown in \citep{blaser2024existential} for a probabilistic model without an underlying SCM and with Fact~\ref{fact:small:model:sat:observed:unobserved} we show that it still holds for Definition~\ref{lab:def:small:model:new:U}.
As soon as we are at the interventional or even counterfactual layer, then the complexity jumps, but only to $\NEXP$. Thus when we have a small model, the problem is hard for a Boolean class, although we have polynomial arithmetic and would have expected that the problem is hard for some theory over the reals.
For linear arithmetic with a compact summation operator, it is known that the complexity increases along the PCH, and we show that the corresponding proofs of \citep{doerflerICLR2025} are still valid if one requires a small-model property. 
The results are summarized in Figure~\ref{fig:smallmodel}.


\section{Satisfiability and validity with requirements on the graph structure of SCMs}
\label{sec:given:DAG:structure}

\subsection{Probabilistic layer}
\label{sec:given:DAG:structure:BNs}
A Bayesian network (BN) $\cB=(\cG,P_{\cB})$ consists of a DAG $\cG$ over nodes representing random variables $X_1,\ldots,X_n$ and a distribution $P_{\cB}$ which factorizes over $\cG$, meaning that the joint probability $\compactEquals{P_{\cB}(x_1,\myldots,x_n)}$ can be written as a product $\prod_{i=1}^n  \compactEquals{P_{\cB}(x_i\mid \pa_i)}$, where $\pa_i$ denotes the values $x_{j_1},\ldots,x_{j_k}$ of parents  $\Pa_i$ of $X_i$ in $\cG$. Then the distribution $P_{\cB}$ is specified as a set of conditional probability tables for each 
variable $X_i$ conditioned on its parents in~$\cG$~\cite{koller2009probabilistic}. So, for example, in the case of binary variables, even if the array representing the joint distribution has $2^n$ non-zero elements, it can still be compactly represented by a set of arrays of total size $\sum_{i=1}^n 2^{|\mbox{\footnotesize \Pa}_{i}|}$. This representation has many additional advantages, for example, enabling efficient probabilistic inference.

In the case of purely probabilistic languages, the problem $\SATprobstargraph$ can be read as follows:
given a DAG $\cG$ over nodes $\bX$ of a Bayesian network (but not being given the distribution $P_{\cB}$) and a formula $\varphi$ 
in the language $\Lprobstar$ decide whether there exists $P_{\cB}$ such that $\varphi$
is true in some BN $\cB=(\cG,P_{\cB})$, where $P_{\cB}$ factorizes over $\cG$.
 
We start with the observation that $\SATprobstar$ is not harder than $\SATprobstargraph$ for any $*\in \Lab$. 
\begin{proposition} \label{prop:DAG:notharder}
\label{fact:sat:sat_with_graph}
    For all  $*\in \Lab$ and 
    for any $\varphi\in \Lprobstar$ over variables $X_1,\ldots,X_n$,
    it is true that $\varphi\in\SATprobstar$ if and only if $(\varphi,\cG_n)\in\SATprobstargraph$,
    where $\cG_n$ denotes a complete DAG over nodes $X_1,\ldots,X_n$,
    with edges $X_i \to X_j$ for all $1\le i < j \le n$.
  \end{proposition}
  \begin{proof}
Assume $\varphi$  is satisfiable in a BN $\cB=(\cG,P_{\cB})$.
Let $\cB'=(\cG_n,P_{\cB'})$, where, 
for every $i$, we define the conditional probability table as
$P_{\cB'}(x_i\mid x_1,\myldots,x_{i-1}):= P_{\cB}(x_i\mid x_1,\myldots,x_{i-1})$
if $P_{\cB}(x_1,\myldots,x_{i-1})>0$; otherwise we set the value arbitrarily, 
e.g., $1/c$, where $c$ is the cardinality of set of values of $X_i$.
Due to the chain rule, we get that 
$P_{\cB'}(x_1,\myldots,x_n)=P_{\cB}(x_1,\myldots,x_n)$
which means that $\varphi$  is satisfiable in  $\cB'$.

The reverse implication is obvious.
\end{proof}

Due to their great importance in probabilistic reasoning, a lot of attention in AI research has been given to both function 
and decision problems, in which the entire 
BN $\cB=(\cG,P_{\cB})$ is given as input with the goal to answer specific queries to the distribution. 
Certainly, to the most basic ones belongs the {\sc BN-Pr} problem to compute,
for a given BN $\cB$ over $\bX$, a variable $X$ in $\bX$, and a value $x$,
the probability $\compactEquals{P_{\cB}(X=x)}$. The next important 
primitive of probabilistic reasoning
consists in finding the Maximum a Posteriori Hypothesis (MAP).
To study its computational complexity, the natural decision problem, named {\sc D-MAP},
has been investigated, which asks if, for a given rational number $\tau$, evidence $\be$, 
and some subset of variables $\bQ$, there is an instantiation $\bq$ to $\bQ$ such that 
$P_{\cB}(\bq,\be) = \sum_{\by} P_{\cB}(\bq,\by,\be)>\tau$.
It is well known that {\sc BN-Pr} is $\sharpP$-complete and 
{\sc D-MAP} is \NP$^{\ccPP}$-complete \cite{roth1996hardness,park2004complexity}.
To relate our setting to the standard Bayesian reasoning, we introduce 
a Model Checking problem,
\begin{definition}
The Bayesian Network Model Checking problem, denoted as $\SATprobstarBN$, 
gets a BN $\cB=(\cG,P_{\cB})$ and a formula $\varphi$ in $\Lprobstar$,
and verifies whether the formula $\varphi$ is satisfied in the BN $\cB$. 
\end{definition}
This problem can be considered as a generalization of the decision version of {\sc BN-Pr}.

\begin{theorem}\label{thm:SATproblinbBN:DAG}
	$\SATprobpolysumBN$ is in $\cP^{\sharpP}$.
	Equivalently, it is in $\cP^{\ccPP}$.
\end{theorem}

Combining this lemma with \cite[Thm. 4]{doerflerICLR2025}, the result 
of \cite{park2004complexity}, and Proposition~\ref{fact:sat:sat_with_graph}, we get the following:
$$
 \SATprobpolysumBN \leqp \mbox{\sc D-MAP} \equiv_ P \SATprobcompsum \equiv_ P  \SATproblinsum  \leqp \SATprobcompsumgraph, 
 $$
which presents the relationships between the model checking,  $\mbox{\sc D-MAP}$, and satisfiability problems
for probabilistic languages  from a complexity perspective.
In the general case, i.e., in the case of polynomial languages with summation, the constraints on DAGs do not make the problems more difficult, but restricting the models to BNs leads to different complexities,  under a standard complexity assumption.

\begin{lemma}\label{prop:prob:comp:sum:graph:er:hard}
$\SATprobcompgraph$ is $ \existsR$-hard.     
\end{lemma}

This means that $\SATprobcompsumgraph$, which is at least as hard as either {\sc D-MAP} or $\SATprobcompgraph$, is both $\NP^{\ccPP}$-hard and $\existsR$-hard. Since these are very distinct classes, one should not expect $\SATprobcompsumgraph$ to be complete for either class. 

\begin{proposition}\label{prop:SATprobstargraph:1:new1}
	$\SATprobpolysumgraph$ is $ \succR$-complete. 
\end{proposition}
\begin{proof}
From  Theorem~\ref{thm:causal_with_graph:reduction}, we know that 
$\SATcausalpolysumgraph$ is $ \succR$-complete.
Since any probabilistic formula is a special case of a counterfactual one and 
$\SATprobpolysum$ is $ \succR$-complete, we get 
$ \SATprobpolysumgraph \leqp \SATprobpolysum$.
On the other hand, from Proposition~\ref{fact:sat:sat_with_graph}, we can conclude that 
the opposite  relation  is also true.
\end{proof}

Thus
 $\SATprobpolysum \equiv_P \SATprobpolysumgraph$
 and the problems are computationally harder than $\SATprobpolysumBN$ if $\cP^{\sharpP} \subsetneq \NEXP .$

\subsection{Interventional and counterfactual reasoning}\label{sec:given:DAG:structure:L2-3}
One of the key components in the development of structural causal models is the 
do-calculus introduced by Pearl \cite{pearl1995causal,Pearl2009}, which allows one 
to estimate causal effects from observational data. The calculus can be expressed
as a \emph{validity} problem for intervention-level languages with the graph structure 
requirements on SCMs. In particular, the rules of the calculus can be seen as 
instances of validity problems for which equivalent graphical conditions in terms 
of $d$-separation are given. For example, the ``insertion/deletion of observations''
rule requires, for a given DAG $\cG$, that, for all SCMs $\fM$ with the causal structure $\cG$,
and for all values $\bx,\by,\bz,\bw$, it holds
$\PP{[\bx]\by \, |\, [\bx](\bz,\bw)}=\PP{[\bx]\by \, |\, [\bx]\bw}$. This rule can be written in a compact way as
$\sum_{\bx,\by,\bz,\bw}(\PP{[\bx]\by \, |\, [\bx](\bz,\bw)}-\PP{[\bx]\by \, |\, [\bx]\bw})^2=0$.
In this section, we prove that, in general, the validity problem with polynomial arithmetic is very hard, namely we will see
that it is complete for ${\tt co}\mbox{-}\succR$.

To this goal, we study the complexity of the satisfiabilities with graph structure 
requirements in the case of the interventional languages.
The base and linear languages are neither weak enough to be independent of the graph structure 
nor strong enough to be able to distinguish graph structures exactly.
In particular the complexity of $\SATinterventcompsumgraph$ and $\SATinterventlinsumgraph$ 
increases above the $\PSPACE$-completeness of $\SATinterventcompsum$ and $\SATinterventlinsum$.
On the other hand, the polynomial languages are strong enough to encode the DAG structure 
of the model which implies that the complexity for these languages is the same 
as $\SATinterventpolysum$.
\begin{theorem}\label{thm:SATinterventcompsumgraph:ccNEXP-complete}
    $\SATinterventcompsumgraph$ and $\SATinterventlinsumgraph$ are $\NEXP$-complete. Thus,
    the validity problems for the corresponding languages are complete for
     ${\tt co}\mbox{-}\NEXP$.
\end{theorem}
An important tool for the proof of this theorem is the satisfiability of a Sch\"{o}nfinkel-Bernays sentence.
The class of Sch\"{o}nfinkel--Bernays  sentences (also called Effectively Propositional Logic, EPR) is a fragment of first-order logic formulas where satisfiability is decidable. Each  sentence in the class is of the form $\exists \bx  \forall \by \psi$ whereby $\psi$ can contain logical operations $\wedge, \vee, \neg$, variables $\bx$ and $\by$, equalities, and relations $R_i(\bx,\by)$ which depend on a set of variables, but $\psi$ cannot contain any quantifier or functions. Determining whether a Sch\"{o}nfinkel-Bernays sentence is satisfiable is an $\NEXP$-complete problem \cite{schoenfinkelLogicNEXPLewis1980} even if all variables are restricted to binary values \cite{schoenfinkelLogicBinaryNEXP2015}.

\begin{corollary}	
	$\SATinterventcompsumgraph$ is computationally harder than  $\SATinterventlinsum$ unless 
	$\NEXP = \PSPACE$.	
\end{corollary}

If the graph $\mathcal{G}$ is complete, that is, all nodes are connected by an edge, 
it is equivalent to a \emph{causal ordering}, an ordering $X_{i_1} \prec X_{i_2} \prec \ldots \prec X_{i_n}$, such that variable $X_{i_j}$ can only depend on variables $X_{i_k}$ with $i_k < i_j$.
If a causal ordering is given as a constraint, it does not change the complexity of the problem:

\begin{lemma}\label{lemm:causal:ordering}
    In $\cL_2$ one can encode a causal ordering.
\end{lemma}
\begin{proof}
Given (in-)equalities in $\cL_2$, we add a new variable $C$ and add to each primitive the intervention $[C=0]$. 
This does not change the satisfiability of (in-)equalities.

Given a causal order $V_{i_1} \prec V_{i_2} \prec \ldots$, 
we add $c$ equations for each variable $V_{i_j}$, $j > 1$:

$\PP{[C=1, V_{i_{j-1}}=k] V_{i_{j}} = k} = 1$ for $k=1,\ldots,c$.

The equations ensure that, if one variable is changed, and $C=1$ is set, the next variable in the causal ordering has the same value, thus fixing an order from the first to the last variable.
\end{proof}

The equivalence between causal orderings and complete graphs thus results in:
\begin{corollary}
$\SATinterventlinsumgraph$ remains  $\PSPACE$-complete in the special case that  $\mathcal{G}$ happens to be a complete graph. 
\end{corollary}

Counterfactual languages on the other hand are strong enough to be able to enforce an explicit graph structure directly in the formula using exponential sums. 
That is, for each variable $X_i$, the formula can encode which variables are the parents of $X_i$ by requiring that $X_i$ remains constant if the parents remain constant. This condition is encoded with an intervention on the parents and an intervention on the other variables, and a counterfactual query ensuring that the second intervention does not change the probability distribution of $X_i$ after the first intervention. An exponential sum can run this through all possible interventions.
\begin{theorem}
    \label{thm:causal_with_graph:reduction}
     $\SATcausalstargraph \equiv_P \SATcausalstar$
    for any $*\in \{\text{base}\langle{\Sigma}\rangle, \text{lin}\langle{\Sigma}\rangle, \text{poly}\langle{\Sigma}\rangle\}$.
\end{theorem}
\onlyFull{
\begin{proof}
    We encode the structure of the graph $\cG$ directly into the formula.

    Let $X_i$ be a variable and let $T_1, \ldots, T_k$ be all the variables that are direct predecessors of $X_i$ in $\cG$.
    We then add the constraint
    \begin{align*}
        \sum_\bv \PP{[t_1, \ldots, t_k]X_i \neqinPP [\bv \setminus x_i] X_i} = 0
    \end{align*}
    to ensure that $X_i$ only depends on $T_1, \ldots, T_k$.
    The sum iterates over all possible values $\bv$ of the endogenous variables.
    Note that this constraint uses a bit of notational sugar, explained in detail in the proof of Theorem~\ref{thm:SATinterventcompsumgraph:ccNEXP-complete}.
\end{proof}}

Since we now know $\SATprobpolysumgraph \leqp \SATinterventpolysumgraph \leqp \SATcausalpolysumgraph \equiv_P \SATcausalpolysum$, Proposition~\ref{prop:SATprobstargraph:1:new1} together with $\SATcausalpolysum$ being $\succETR$ complete \cite{doerflerICLR2025} we obtain:
\begin{corollary}
    $\SATinterventpolysumgraph$ is $\succETR$ complete.
\end{corollary}

\section{Probabilistic and causal reasoning with a small model}\label{sec:etr:sum}
\label{sec:small}

In~\citep{blaser2024existential}, the complexity of a problem named $\SATprobpolysumsm$ is investigated. 
The authors consider purely probabilistic models where $sm$ means $$\#\{(x_1,\ldots,x_m): \mathbb{P}(X_1=x_1,\ldots, X_n=x_m) > 0\}$$ is small, i.e., is polynomially bounded in the input size. 
Before talking further about $\SATprobpolysumsm$, we need to show that their problem $\SATprobpolysumsm$ is equivalent to our problem $\SATprobpolysumsm$ according to Definition~\ref{lab:def:small:model:new:U} where $$\#\{(u_1,\ldots,u_m): P(U_1=u_1,\ldots, U_m=u_m) > 0\}$$ is small (we cannot use their definition, as it is not meaningful to give constraints on the distribution $\mathbb{P}(\bX)$ for \emph{causal} models where interventions can change this distribution in various ways).  

\newcommand{\countPosU}[1]{\#_\bU^+(#1)}
\newcommand{\countPosX}[1]{\#_\bX^+(#1)}

\begin{fact}\label{fact:small:model:observed:below:unobserved}
For a given model $\fM$, let $\countPosX{\fM} = \#\{(x_1,\ldots,x_n): \PP{X_1=x_1,\ldots, X_n=x_n} > 0\}$ be the number of assignments to the observed variables with positive probability and $\countPosU{\fM}= \#\{(u_1,\ldots,u_m): P(U_1=u_1,\ldots, U_m=u_m) > 0\}$ be the number of assignments to the unobserved variables with positive probability.

Then 
$\countPosX{\fM} \leq \countPosU{\fM}$.
\end{fact}
\begin{proof}
The functions $\cF$ map an unobserved assignment $\bu$ to an observed assignment $\bx$ deterministically. 
Thus any assignment $\bu$ with positive probability $P(\bu)$ can only make at most a single value $\PP{\bx}$ non-zero.
\end{proof}
The reverse does not hold. All functions in the model might be constant, in which case $\#_\bX^+ = 1$, regardless of $\#_\bU^+$. 
Thus a model that is a solution for the $\SATprobpolysumsm$ problem of \citep{blaser2024existential} might not be a solution to our problem $\SATprobpolysumsm$.
Still:

\begin{fact}\label{fact:small:model:sat:observed:unobserved}
For a given probabilistic formula $\varphi$ 
and a number $p \in \IN$,
there exists a semi-Markovian model  $\fM=(\cF, P, \bU=\{U_1,\ldots,U_m\}, \bX)$ 
	such that $\fM \models\varphi$ and
    $\countPosU{\fM} \le p$
    if and only if 
there exists a model  $\fM'=(\cF', P', \bU'
, \bX)$ 
	such that $\fM' \models\varphi$ and
    $\countPosX{\fM'} \le p$.
\end{fact}
\onlyFull{
\begin{proof}
"$\Rightarrow{}$": If there exists a model $\fM$, we can set $\fM'=\fM$ due to Fact~\ref{fact:small:model:observed:below:unobserved}.

"$\Leftarrow{}$": If there exists a model $\fM'$, we define the model $\fM$ as follows:

\begin{compactitem}
\item $\bU = (U_1,\ldots,U_n)$,
\item $F_i(u_i) = u_i$,
\item $P(U_1=u_1,\ldots,U_n = u_n) = \mathbb{P}'(X_1=u_1,\ldots, X_n=u_n)$.
\end{compactitem}
Then $\mathbb{P} = P = \mathbb{P}'$ due to identifying $U_i$ with $X_i$. Thus $\fM$ and $\fM'$ are indistinguishable for all probabilistic formulas.
\end{proof}
}

Thereby it is critical to consider semi-Markovian models, as the unobserved variables constructed in the proof might not be independent, as it was required in Markovian models. 

\begin{fact}\label{fact:small:model:sat:not:Markovian}
Fact~\ref{fact:small:model:sat:observed:unobserved} does not hold for Markovian models. 
\end{fact}

Nevertheless, Fact~\ref{fact:small:model:sat:observed:unobserved} implies that the $\existsR^\Sigma$-completeness proof for a ``$\SATprobpolysumsm$'' version restricted on the observed distribution in \citep{blaser2024existential} also applies to our $\SATprobpolysumsm$. 
Thus we know that $\SATprobpolysumsm$ is complete for a new complexity 
class $\existsR^\Sigma$, a class that extends the existential theory of reals with summation operators. 
This is not a succinct class, so it is different from $\succR$, in fact they show
$$\existsR \cup \NP^{\ccPP}  \subseteq \existsR^\Sigma \subseteq \PSPACE \subseteq \NEXP \subseteq \succR = \realNEXP,$$
where $\realNEXP$ is the class of problems decided by nondeterministic real 
Random Access Machines (RAMs) in  
exponential time (see   \citep{blaser2024existential} for the exact details).

Not allowing multiplication reduces the complexity, however, once we allow for interventions, the complexity increases again. We know the following complexities for linear arithmetic from \citep{doerflerICLR2025}, whose results translate to small models:
 
\begin{lemma}
    \label{lemma:small:model:in:linear}
    The following completeness results hold for small models:
    \begin{compactitem}
        \item $\SATprobcompsumsm$ and $\SATproblinsumsm$ are $\NP^{\ccPP}$-complete.
        \item $\SATinterventcompsumsm$ and $\SATinterventlinsumsm$ are $\PSPACE$-complete.
        \item $\SATcausalcompsumsm$ and $\SATcausallinsumsm$ are $\NEXP$-complete.
    \end{compactitem}
\end{lemma}

Allowing multiplications again leads to $\NEXP$-completeness, even on the second (interventional) level.
Note however that this is likely still weaker than the $\realNEXP$-completeness 
of $\SATinterventpolysum$ and $\SATcausalpolysum$.
\begin{theorem}\label{lab:lemma:interv:counter:smconj}
	$\SATinterventpolysumsm$ and $\SATcausalpolysumsm$ are $\NEXP$-complete.
\end{theorem}

This follows because a causal model consists of a probabilistic part $\PP{\bu}$ and a deterministic, causal part $\cF$.
In the small model, $\PP{\bu}$ is restricted to have polynomial size, but $\cF$ can still have exponential size.
In $\SATprobpolysumsm$, one cannot reason about $\cF$, but $\SATinterventpolysumsm$ and $\SATcausalpolysumsm$ can access the full power of $\cF$. Since $\NEXP$ subsumes $\existsR$, reasoning about polynomially many real numbers (probabilities) cannot increase the complexity.

\section{Discussion}\label{sec:landscape}

We investigated the computational complexities of probabilistic satisfiability problems from a multi-parametric point of view. Our new completeness results nicely extend and complement the previous achievements by \cite{fagin1990logic, ibeling2022mosse, zander2023ijcai, blaser2024existential}. Our main focus was on the effect of constraining the model. 
We have shown that including a DAG in the input increases the complexity on the second level for linear languages (i.e., $\SATinterventlinsumgraph$) to \NEXP-complete, while it does not change the complexity on the third level or for the full languages. We leave the exact complexity of the first linear level (i.e., $\SATproblinsumgraph$) for future work, but we conjecture that it is  $\existsR^\Sigma$-complete.
On the other hand, including a Bayesian network, a DAG and its probability tables reduces the complexity to being in $\cP^{\ccPP}$. Here, the completeness is left for future work\footnote{As the main computation in the proof of Lemma~\ref{thm:SATproblinbBN:DAG} is performed by the \ccPP-oracle; one can extend the proof to show the containment of $\SATprobpolysumBN$ in the space-limited complexity class $\mathsf{L}^\ccPP$ and in the circuit complexity class $\mathsf{NC}_1^{\ccPP}$. Thus, it is unlikely to be $\cP^{\ccPP}$-complete, although it might be possible that all three classes are identical.}.

Another interesting feature is that when the model is unconstrained, then the completeness  for linear 
languages is expressed in terms of standard Boolean classes while 
the completeness of satisfiability for languages involving polynomials 
over the probabilities requires classes over the reals. However, once we constrain the model to be small, then the complexity with polynomial arithmetic is again described by a Boolean complexity class (for PCH layers two and three). This is due to the fact that  the running time of Renegar's algorithm is essentially determined by the number of variables (which is small in the case of a small model) and not by the number of equations (which still can be large).


\bibliography{lit-icalp} 

\onlyProc{\end{document}}

\newpage
\appendix


\section{Constrained models with expanded marginalization}

\label{app:expanded}

Moss{\'e} et al.\ \cite{ibeling2022mosse} investigate the complexity of satisfiability with different arithmetics (basic, linear, polynomial) and layers (probabilistic, interventional, counterfactual) of the PCH. They do not have a compact marginalization operator. In this setting, they prove that only the arithmetic matters. If the arithmetic is basic or linear, then all satisfiability problems are $\NP$-complete, independent of the PCH layer. For the probabilistic layer, this was also proven by Fagin et al. \cite{fagin1990logic}. If the arithmetic is polynomial, then the satisfiability problems are $\existsR$-complete. Moss{\'e} et al.\ in particular prove a small model property \cite[Lemma 4.7]{ibeling2022mosse}: In each case, if there is a model, then there is one with only polynomially many entries with non-zero probability in the joint distribution. That means that requiring that the model is small does not change the complexity of the problem. However, the proof of Lemma 4.7 by Moss{\'e} et al.\ even yields a stronger statement. They show that for every fixed graph structure of the SCM, if there is a model with this graph structure, then there is also a small model with the same graph structure. Therefore, constraining the graph does not change the complexity of the respective satisfiability problem.

\section{Technical Details and Proofs}\label{sec:appendix}

\subsection{Proofs of Section \ref{sec:given:DAG:structure}}

\begin{proof}[Proof of Lemma~\ref{thm:SATproblinbBN:DAG}]
To show that  $\SATprobpolysumBN$ is in $\cP^{\sharpP}$,
assume $\varphi$ is a probabilistic formula with polynomial terms
and  $\cB=(\cG,P_{\cB})$ a BN over variables $X_1,\ldots,X_n$.
We modify $\varphi$ to an equivalent formula in which all primitives are 
of the form $\PPB{X_1 = \xi_1,\dots,X_n = \xi_n}$, that is, 
every variable occurs in the primitive and the basic events are connected by conjunctions.   
This is achieved inductively, similar to the proof of \cite[Fact 2]{doerflerICLR2025}:
For a given $\PPB{\delta}$, 
we define the arithmetization $A_\delta(\xi_1,\dots,\xi_n)$ of $\delta$ to be $1$
if $\PPB{X_1 = \xi_1,\dots,X_N = \xi_n}$
contributes to $\PPB{\delta}$, and $0$ otherwise. 
$A_\delta$ can be defined inductively by following the
logical structure of $\delta$. We first present the construction
for the Boolean domain for simplicity.
    \begin{itemize}
        \item For the base case $\delta$ equals to $X_i = \xi$, the arithmetization is 
        $A_{\delta}(X_1,\dots,X_n)=1-(X_i - \xi)^2$. 
        \item If $\delta = \delta_1 \wedge \delta_2$, then 
        $A_\delta = A_{\delta_1} A_{\delta_2}$.
        \item If $\delta = \neg \delta_1$, then $A_\delta = 1 - A_{\delta_1}$.
    \end{itemize}
    Since we implement $A_\delta$
    as a polynomial, it is also easy to implement the basic 
    predicates ``$X_i = \xi$'' over larger domains $\mathit{Val}$.
    It is simply a polynomial of degree $|\mathit{Val}|$ that is $1$ on $\xi$
    and $0$ everywhere else. Such a polynomial can be 
    found using Lagrange interpolation.
    
    Then we have 
    \begin{equation}\label{eq:pspace:estimation:of:t}
        \PPB{\delta} = \sum_{\xi_1} \dots \sum_{\xi_n} A_\delta(\xi_1,\dots,\xi_n)
        \PPB{X_1 = \xi_1,\dots,X_n = \xi_n}.
    \end{equation}
where $\PPB{X_1 = \xi_1,\dots,X_n = \xi_n}=\prod_{i=1}^n 
\PPB{X_i=\xi_i \mid \mbox{\it pa}_i}$ is computed based on the 
conditional probability tables.
Let $D$ be the product of all denominators of conditional probabilities 
(given as rational numbers) in all tables
of $\cB$. Obviously, the size of the representation of $D$ is bounded polynomially 
in the size of the input $\cB$. 
We will compute the probabilities of the basic terms $\PPB{X_1 = \xi_1,\dots,X_n = \xi_n}$
in form $n_{\xi_1,\dots, \xi_n}/D$
for integers $n_{\xi_1,\dots, \xi_n}$ with 
$0\le n_{\xi_1,\dots, \xi_n}\le D$.
Consequently, the value  $\PPB{\delta}$ will be represented as $n_{\delta} /D$.

To decide if $\varphi$ is true, we estimate the numerical values of terms
in any inequality $\bt_1\le  \bt_2$ of~$\varphi$. Since subtraction (and negation) 
is allowed, we will show how to verify that $\bt> 0$, for a polynomial term $\bt$.
Moreover, w.l.o.g., we will assume the negative signs are pushed to the 
bottom of the expression, i.e., $\bt$ uses addition, multiplication, and summation 
operators applied on basic terms $ s\cdot \PPB{\delta}$, with $s\in\{+1,-1\}$.
To decide if $\bt> 0$, we use a $\ccPP$ algorithm, described as 
Algorithm~\ref{alg:NTM:counting}, which, on input  $\bt,  \cB$,  generates  
a nondeterministic computation tree with more accepting than rejecting paths 
if and only if $\bt> 0$.
Below we show that the algorithm works correctly.

Let $\bt$ be  a term of degree $\Delta=\deg(\bt)$
and let $\bt=\sum_j \prod_{p_{\delta}\in J_j} p_{\delta}$ be an expansion of the term, 
i.e., let each $p_{\delta}$ be an 
atomic term $\PPB{\delta}$ or  $-\PPB{\delta}$, for some $\delta$.
We partition the set $\cJ_{\bt}=\{J_1,J_2,\ldots \}$ into $\cJ_{\bt}^{\Plus}\cup\cJ_{\bt}^{\Minus}$ 
such that,  for any $J\in \cJ_{\bt}^{\Plus}$, the product $\prod_{p_{\delta}\in J} p_{\delta}>0$
and,  for any $J\in \cJ_{\bt}^{\Minus}$, we have $\prod_{p_{\delta}\in J} p_{\delta}<0$.
Since every $\PPB{\delta}$ is represented as $n_{\delta} /D$, we can rewrite the 
expansion as 
\begin{align*}
\bt\ &=
	\sum_{J\in \cJ_{\bt}^{\Plus}} \frac{1}{D^{|J|}}\prod_{p_{\delta}\in J} n_{\delta} +
	\sum_{J\in \cJ_{\bt}^{\Minus}} \frac{1}{D^{|J|}}\prod_{p_{\delta}\in J} n_{\delta} \\[1mm]
      &=
	 \frac{1}{D^{\Delta}} \sum_{J\in \cJ_{\bt}^{\Plus}}  D^{{\Delta}-|J|} \prod_{p_{\delta}\in J}  n_{\delta} +
	 \frac{1}{D^{\Delta}} \sum_{J\in \cJ_{\bt}^{\Minus}} D^{{\Delta}-|J|} \prod_{p_{\delta}\in J} n_{\delta} .
\end{align*}
Thus, $\bt>0$ if 
$ \sum_{J\in \cJ_{\bt}^{\Plus}}  D^{{\Delta}-|J|} \prod_{p_{\delta}\in J}  n_{\delta} >
 \sum_{J\in \cJ_{\bt}^{\Minus}} D^{{\Delta}-|J|} \prod_{p_{\delta}\in J} n_{\delta}$.
 
To see that the algorithm correctly estimates that this inequality holds,
we will only take into account the accepting/rejecting paths that are determined 
in Line~\ref{counting:line:accept} and \ref{counting:line:reject} of the algorithm.
The paths generated in Line~\ref{counting:line:acc:rej} are irrelevant because 
they have no impact on the final result.
By structural induction, we prove that \emph{Count}($\compactEquals{\bt, S, d,\sigma}$),
with empty $S$, $d\ge\Delta-1$, and $\sigma=+1$, generates
$\sum_{J\in \cJ_{\bt}^{\Plus}} D^{d+1-|J|}\prod_{p_{\delta}\in J} n_{\delta}$
accepting  and 
$\sum_{J\in \cJ_{\bt}^{\Minus}} D^{d+1-|J|}\prod_{p_{\delta}\in J} n_{\delta}$
rejecting paths. Thus, calling \emph{Count} with $\bt$ and 
$d=\Delta-1$, we get the correct result.

If $\bt=\PPB{\delta}$, then \emph{Count}($\compactEquals{\bt, S, d,\sigma}$),
generates $n_{\delta}$ accepting paths and  
no rejecting paths. If $\bt=-\PPB{\delta}$, then the resulting computation tree has  
$n_{\delta}$ rejecting paths and no accepting ones.
For $\bt=\bt_1 +\bt_2$, from inductive assumption, we have that 
the procedure generates
$\sum_{i=1}^2\sum_{J\in \cJ_{\bt_i}^{\Plus}} D^{\Delta-|J|}\prod_{p_{\delta}\in J} n_{\delta}=
\sum_{J\in \cJ_{\bt}^{\Plus}} D^{\Delta-|J|}\prod_{p_{\delta}\in J} n_{\delta}$
accepting  and 
$\sum_{i=1}^2\sum_{J\in \cJ_{\bt_i}^{\Minus}} D^{\Delta-|J|}\prod_{p_{\delta}\in J} n_{\delta}=
\sum_{J\in \cJ_{\bt}^{\Minus}} D^{\Delta-|J|}\prod_{p_{\delta}\in J} n_{\delta}$
rejecting paths. 
Case $\bt=\mbox{$\sum_{x} \bt'$}$ is analogous to common summation.

Now, let  $\bt=\bt_1 \cdot \bt_2$. Note that we can actually assume a slightly stronger 
inductive assumption, namely that after pushing $\bt_2$ to the stack $S$ and 
calling \emph{Count}($\bt_1, S, d,\sigma$) in Line~\ref{counting:line:push:product},
the algorithm, until popping~$\bt_2$ from the stack $S$ in Line~\ref{counting:line:empty:stack},  
generates 
$\sum_{J\in \cJ_{\bt_1}^{\Plus}} \prod_{p_{\delta}\in J} n_{\delta}+\sum_{J\in \cJ_{\bt_i}^{\Minus}} \prod_{p_{\delta}\in J} n_{\delta}$
paths. For every path, the algorithm stores the sign $\sigma=+1$
if the path corresponds to the product $J\in \cJ_{\bt_1}^{\Plus}$
and $\sigma=-1$ for  $J\in \cJ_{\bt_1}^{\Minus}$. Next 
it pops $\bt_2$ from the stack $S$, making it empty, and, 
with $d=\Delta-1-|J|$, it calls
\emph{Count}($\bt_2, S, d,\sigma$) in Line~\ref{counting:line:empty:stack:call}.
By the inductive assumption, we get that finally the algorithm generates 
\begin{align*}
&\mbox{$
\sum_{J\in \cJ_{\bt_1}^{\Plus}} \prod_{p_{\delta}\in J} n_{\delta}\cdot 
\left(\sum_{J'\in \cJ_{\bt_2}^{\Plus}} D^{\Delta+|J|+|J'|}\prod_{p_{\delta}\in J'} n_{\delta}\right) $}\\
&\quad \quad + 
\mbox{$
\sum_{J\in \cJ_{\bt_1}^{\Minus}} \prod_{p_{\delta}\in J} n_{\delta}\cdot 
\left(\sum_{J'\in \cJ_{\bt_2}^{\Minus}} D^{\Delta+|J|+|J'|}\prod_{p_{\delta}\in J'} n_{\delta}\right) 
$}
\end{align*}
accepting paths and 
\begin{align*}
&\mbox{$
\sum_{J\in \cJ_{\bt_1}^{\Plus}} \prod_{p_{\delta}\in J} n_{\delta}\cdot 
\left(\sum_{J'\in \cJ_{\bt_2}^{\Minus}} D^{\Delta+|J|+|J'|}\prod_{p_{\delta}\in J'} n_{\delta}\right)$}\\
&\quad \quad + 
\mbox{$
\sum_{J\in \cJ_{\bt_1}^{\Minus}} \prod_{p_{\delta}\in J} n_{\delta}\cdot 
\left(\sum_{J'\in \cJ_{\bt_2}^{\Plus}} D^{\Delta+|J|+|J'|} \prod_{p_{\delta}\in J'} n_{\delta}\right) 
$}
\end{align*}
rejecting paths. This completes the proof.

\begin{algorithm}[ht]
\SetAlgoLined
\KwIn{ 
$\Tprobpolysum$ term $\bt$ and BN $\cB=(\cG,P_{\cB})$.} 
\KwOut{Nondeterministic computation tree with more accepting than rejecting paths
iff $\bt > 0$.}

 Set stack $S$ initially empty; Let $d:=\deg(\bt)-1$, and let $\sigma:=+1$\;
 Call \emph{Count}$(\bt, S, d,\sigma)$\;
\vspace{0.2cm}
\SetKwProg{myfunc}{Procedure}{}{}
\myfunc{Count$(\bt, S, d,\sigma)$}
{
\KwIn{Current term $\bt$; stack of terms to be evaluated $S$; integer $d$; sign~$\sigma\in\{+1,-1\}$}

\Switch{according to the structure of $\bt$}
{
	\Case{$\bt=\bt_1 +\bt_2$}{ 
		Split non-deterministically into two branches:\label{counting:line:t1:plus:t2:split}\\
		\quad in the first one call \emph{Count}($\bt_1, S, d,\sigma$) and \\
		\quad in  the second one call \emph{Count}($\bt_2, S, d,\sigma$);
	}

	\Case {$\bt=\bt_1 \cdot\bt_2$}{ 
		Push $\bt_2$ to stack  $S$ and 
	 	call \emph{Count}($\bt_1, S, d,\sigma$)\label{counting:line:push:product}\;
	}
	\Case{$\bt=\mbox{$\sum_{x} \bt'$}$}{ 
		Split non-deterministically into $|\mathit{Val}|$ branches, 
		each one for a single value $\xi$ of the index variable $x$\;
		 At the end of each branch call 
		 \emph{Count}($\bt'[\xi /x], S, d,\sigma$)\;
	 }	
	\Case{$\bt=s\cdot \PPB{\delta}$ is the primitive term with sign $s\in\{+1,-1\}$}{
	Split non-deterministically $n$ times
    	into $|\mathit{Val}|$ branches and store the selected values $\xi_1, \xi_2, \ldots,\xi_n$\;
    	Compute $A_{\delta}(\xi_1,\ldots,\xi_n)$\;
	Compute (from the conditional probability tables $P_{\cB}$) integer $n_{\xi_1,\dots, \xi_n} $  such that
	$\PPB{X_1 = \xi_1,\dots,X_n = \xi_n}=n_{\xi_1,\dots, \xi_n} / D$\;
    	\eIf {$A_{\delta}(\xi_1,\ldots,\xi_n)=0$ or $n_{\xi_1,\dots, \xi_n} =0$}{ 
	Split into one rejecting and one accepting branch and terminate\label{counting:line:acc:rej}\;
	}	
	{	\eIf{$S$ empty}{ 
    		Generate, in a nondeterministic way, 
		$n_{\xi_1,\dots, \xi_n} \cdot D^{d}$ computation paths;\\
		If $\sigma\cdot s =+1$, then set all of them accepting,\label{counting:line:accept}\\
	 	\quad otherwise, set all of them rejecting\label{counting:line:reject}\;}
		{Generate, in a nondeterministic way, $n_{\xi_1,\dots, \xi_n} $ computation paths\;		
		At the end of each path {\bf do}\\
		\quad Set $d:=d-1$ and $\sigma:=\sigma\cdot s$\;
		\quad Pop $\bt$ from $S$\label{counting:line:empty:stack}\;
		\quad Call \emph{Count}($\bt, S, d,\sigma$)\label{counting:line:empty:stack:call}\;
	}
	}
	}
}
}
    \caption{PP Algorithm to decide if $\bt > 0$.}\label{alg:NTM:counting}
\end{algorithm}

\end{proof}

\newpage

\begin{proof}[Proof of Lemma~\ref{prop:prob:comp:sum:graph:er:hard}]
\newcommand{\EQSEP}{;\ \ }
To show the hardness, we will reduce the $\existsR$-complete problem ETR-INV to $\SATprobcompsumgraph$ with a binary domain. An ETR-INV instance consists of three types of equations~\citep{abrahamsen2018art}:
$$x_i = 1\EQSEP{}x_i + x_j = x_k\EQSEP{}x_i \cdot x_j = 1$$
where $x_i,x_j,x_k\in \{x_1,\ldots,x_n\}$ and asks if there exist variables $x_1,\ldots,x_n\in [1/2, 2]$ that satisfy all equations.
We double each variable to reduce the range to $[1/4,1]$, which fits in a probability distribution, and obtain equivalent equations:
$$2x_i = 1\EQSEP{}2x_i +2 x_j = 2x_k\EQSEP{}2x_i \cdot 2x_j = 1$$
which are again equivalent to
$$x_i = 1/2\EQSEP{}x_i + x_j = x_k\EQSEP{}x_i \cdot x_j = 1/4.$$

We introduce random variables $X_1,\ldots,X_n$, and replace each value $x_i$ with $\PP{X_i = 0}$, obtaining equations
$$\PP{X_i = 0} = 1/2\EQSEP{}\PP{X_i = 0} + \PP{X_j = 0} = \PP{X_k = 0}\EQSEP{}\PP{X_i = 0} \cdot \PP{X_j = 0} = 1/4.$$
Additionally, we need the constraint $$\PP{X_i = 0} \ge 1/4$$ to encode $x_i\in[1/4,1]$.

This already gives us---trivially---that $\SATprobpoly$ is $\existsR$-hard.

Now we remove the multiplication, by cloning each random variable $X_i$ to a variable $Y_i$ with the constraint $$\PP{X_i=0}=\PP{Y_i=0},$$ and replacing them on the right side of the multiplication:
$$\PP{X_i = 0} = 1/2\EQSEP{}\PP{X_i = 0} + \PP{X_j = 0} = \PP{X_k = 0}\EQSEP{}\PP{X_i = 0} \cdot \PP{Y_j = 0} = 1/4.$$
We can force all variables to be independent by using a graph $\cG$ without edges, which allows us to write the equations as:
$$\PP{X_i = 0} = 1/2\EQSEP{}\PP{X_i = 0} + \PP{X_j = 0} = \PP{X_k = 0}\EQSEP{}\PP{X_i = 0, Y_j = 0} = 1/4.$$

This already gives us that $\SATproblingraph$ is $\existsR$-hard.

Now we remove the addition by introducing two new variables $S_{ij}$ and $S'_{ij}$ for each addition, where $S'_{ij}$ is a child of $S_{ij}$ in the graph. We use constraints $$\PP{S_{ij} = 0, S'_{ij} = 0} = \PP{X_i = 0}\EQSEP{}\PP{S_{ij} = 1, S'_{ij} = 0} = \PP{X_j = 0}.$$ 
Because $\sum_{t} \PP{S_{ij} = t, S'_{ij} = 0}=\PP{S'_{ij} = 0}$ and the value should not be larger than $\PP{X_k = 0}$, this gives us 
$$\PP{X_i = 0} = 1/2\EQSEP{}\PP{S'_{ij} = 0}  = \PP{X_k = 0}\EQSEP{}\PP{X_i = 0, Y_j = 0} = 1/4.$$

This gives us that $\SATprobcompgraph$ is $\existsR$-hard\footnote{The two remaining constants can be obtained as usual through temporary variables from $\PP{\top} = 1$ and the addition reduction.}.
\end{proof}

\begin{proof}[Proof of Theorem~\ref{thm:SATinterventcompsumgraph:ccNEXP-complete}]

    We follow along with the proof of \cite[Thm. 9]{doerflerICLR2025} showing that $\SATcausalcompsum$ is $\NEXP$ hard.
    We reduce from the satisfiability of Schönfinkel-Bernay sentences.
    Boolean values are represented as the value of the random variables, with $0$ meaning \false{} and $1$ meaning \true{}.
    We will assume that $c=2$, so that all random variables are binary, i.e. $\mathit{Val} = \{0,1\}$.
    
    We follow the convention of \cite{doerflerICLR2025} and write (in)equalities between random variables as $=$ and $\neq$. In the binary setting, $X=Y$ is an abbreviation for $(X = 0 \land Y = 0) \vee (X = 1 \land Y=1)$, and $X\neq Y$ an abbreviation for $\neg (X=Y)$.
    To abbreviate interventions, we write $[w]$ for $[W=w]$, $[\bw]$ for interventions on multiple variables $[\bW=\bw]$, and $[\bv\setminus\bw]$ for interventions on all endogenous variables except $\bW$.
    
    Random variables  
    $\bX=\{X_1,\ldots X_n\}$ and $\bY= \{Y_1,\ldots Y_n\}$ are used for the quantified Boolean variables $\bx,\by$
    in the Schönfinkel-Bernay sentence $\exists \bx  \forall \by \psi$.
    For each distinct $k$-ary relation $R_i(z_1,\ldots,z_k)$ in the formula, we define a random variable $R_i$ and variables $Z^1_i,\ldots,Z_i^k$ for the arguments. For the $j$-th occurrence of that relation $R_i(t^1_{ij},\ldots,t^k_{ij})$ with $t^l_{ij} \in \{x_1,\ldots,x_n,y_1,\ldots,y_n\}$, we define another random variable $R^j_i$. 
    
    So far the setup is identical to \cite{doerflerICLR2025}, however for $\SATinterventcompsumgraph$, we are not able to use counterfactual reasoning.
    Instead, we have to use the graph structure to encode most of our constraints.

    
    We use the graph $\cG$ to ensure that $R_i$ only depends on its arguments.
    For this, the incoming edges to $R_i$ are exactly the edges from $Z_i^1, \ldots, Z_i^k$.
    In the same way we ensure that $R^j_i$ only depends on its arguments by making the only incoming edges to $R^j_i$ exactly the edges from $T_1, \ldots, T_k$.
    
    We add the following constraint to ensure that $R^j_i$ and $R_i$ have an equal value for equal arguments:
    \begin{equation}\label{eqn:proof:lin:intervent:nexptime:consistent:relations}
    \sum_{t_1,\ldots,t_k} \PP{[T_1=t_1,\ldots,T_k=t_k, Z^1_i = t_1,\ldots,Z^k_i=t_k] (R^j_i \neqinPP R_i)} = 0 .
    \end{equation}
    
    
    We continue by ensuring that the values of $\bX$ are not affected by the values of $\bY$ by making the $\bY$ the only predecessors of the $\bX$ in $\cG$.
     
    Let $\psi'$ be obtained from $\psi$ by replacing equality and relations on the Boolean values with the corresponding definitions for the random variables:
    \begin{equation} \label{eqn:proof:lin:intervent:nexptime:main:equation}
    \sum_{\by} \PP{[\by] \psi' } = 2^n
    \end{equation}
    
    Suppose the $\SATinterventcompsumgraph$ instance is satisfied by some model $\fM$.
    We need to show that $\exists \bx  \forall \by \psi$ is satisfiable.
    First note that the values $\bx$ of the variables $\bX$ deterministically depend only on the values $\bu$ of the exogenous variables.
    Choose any value $\bu$ (and thus corresponding $\bx$) that contributes a positive probability to any of the $\PP{[\by] \psi' }$ in equation \eqref{eqn:proof:lin:intervent:nexptime:main:equation}.
    If there was any $\bx$ chosen this way that would not satisfy $\psi$ for all values of $\by$, $\PP{[\by] \psi' }$ would be less than $ 1$  for these values of $\bx$ (determined by $\bu$) and $\by$, and equation \eqref{eqn:proof:lin:intervent:nexptime:main:equation} would not be satisfied.
    The relations $R_i$ are chosen as the values of the random variables $R_i$.
    Equation \eqref{eqn:proof:lin:intervent:nexptime:consistent:relations} now enforces two things, together with the graph structure.
    First of all, the $R^j_i$ and $R_i$ depend deterministically only on their parameters and possibly the $\bu$ values.
    Since we have fixed the values $\bu$, they only depend on their parameters for the rest of this proof.
    Secondly, all occurrences of $R_i$ are given the same value.
    The satisfiability of $\PP{[\by] \psi' }$ thus implies the satisfiability of $\psi$ and ultimately of $\exists \bx  \forall \by \psi$.
    
    Suppose $\exists \bx  \forall \by \psi$ is satisfiable.
    Create a deterministic model $\fM$ by choosing the values of the random variables $\bX$ to the values chosen by $\exists \bx$.
    The random variables $R_i$ (and $R^j_i$) are chosen as functions depending on their parameters, behaving exactly like the value of the relation $R_i$.
    These are allowed dependencies according to our constructed graph structure.
    Clearly equation \eqref{eqn:proof:lin:intervent:nexptime:consistent:relations} is satisfied by this.
    The remaining random variables (the $\by$) can be chosen arbitrarily as constants.
    Their values do not matter since they are overridden by interventions whenever their value would become relevant.
    It remains to consider Equation~\eqref{eqn:proof:lin:intervent:nexptime:main:equation}, but this one is satisfied because $\psi$ is satisfied for all values of $\by$.
    
    Additionally, $\SATinterventcompsumgraph$ is a special case of $\SATcausalcompsumgraph$ and thus can be solved in $\NEXP$ by Theorem~\ref{thm:causal_with_graph:reduction} and \cite[Thm. 9]{doerflerICLR2025}. 
\end{proof}

\subsection{Proofs of Section \ref{sec:etr:sum}}
 
\def\setofinterventions{{\boldsymbol{\alpha}}}
\begin{algorithm}[ht]
    \SetAlgoLined
    \KwIn{$\Lcausalpolysum$ formula $\varphi$ and a unary number $p \in \IN$}
    \KwOut{Is the instance satisfiable with a model of size at most $p$?}

    Guess a causal order $X_1,\ldots,X_n$\;
    Explicitly expand all exponential sums in $\varphi$\;

    For each probability $\PP{\delta_i}$ occurring in $\varphi$ introduce a variable $P_i$, initialized to $0$\;
    Introduce symbolic variables $z_1, \ldots, z_p$\;
    \For{$i \in \{1, \ldots, p\}$}{
        Guess a deterministic SCM $\fM$ using the causal order $X_1, \ldots, X_n$.
        \For{each $\PP{\delta_i}$ that is satisfied by $\fM)$}{
            Symbolically increment $P_j$ by $z_i$.
        } 
    }
    Replace all probabilities $\PP{\delta_i}$ by the values of $P_i$ and check whether the resulting $\ETR$-formula (in the variables $z_1, \ldots, z_p$) is satisfiable using Renegar's algorithm\;
    \caption{Solving $\SATcausalpolysumsm$ in $\NEXP$}\label{alg:causal:polysm:sum:in:NEXP}
\end{algorithm}

\begin{proof}[Proof of Fact~\ref{fact:small:model:sat:not:Markovian}]
Let $p = 3$, $n = 2$, and the domain of $\bX$ be binary.

We give a counter example of a model $\fM$ where $\countPosX{\fM}\le p$, but no equivalent Markovian model $\fM'$ with $\countPosU{\fM'} \le p$ exists.   


Let $P(u_1, u_2) = 1/4$ for all $u_1,u_2$. Let $F_1(u_1) = u_1$ and $F_2(x_1, u_2) = u_2 (1 - x_1)$.

For easier readability, this model can also be shown as a table:
$\begin{array}{cc|c|cc}
U_1 & U_2 & P(\bu) & X_1 & X_2 \\
\hline\hline
0&0&1/4&0 & 0\\
0&1&1/4&0 & 1\\
1&0&1/4&1 & 0 \\
1&1&1/4&1 & 0\\
\end{array}$.


Then $\mathbb{P}(x_1,x_2) = \left(\begin{array}{cc}
  1/4   & 1/2 \\
  1/4   &  0
\end{array}\right)$, where the column depends on the value for $X_1$ and the row on the value for $X_2$.

Then $\fM$ is a small model for the observed variables with $\#_\bX^+ = 3 \leq p$.

The models $\fM$ and $\fM'$ are only equivalent if $3 = \countPosX{\fM} = \countPosX{\fM'} $.
Due to Fact~\ref{fact:small:model:observed:below:unobserved}, $\countPosX{\fM'} \leq \countPosU{\fM'}$.
If $\countPosU{\fM'} > p = 3$, it is not a small model according to Definition~\ref{lab:def:small:model:new:U}. Thus $\countPosU{\fM'} = 3$.

By definition, a Markovian model $\fM'$ for two observed variables must have two unobserved variables $U_1, U_2$ that are independent of each other. Then $\#_\bU^+$ is the product of the ranges of the unobserved variables, i.e., $\#_\bU^+(\fM') = \#\{u_1: P(U_1=u_1) > 0\} \cdot \#\{u_2: P(U_2=u_2) > 0\}$.
But $3$ is a prime number, which means there is one unobserved variable taking three values and another unobserved variable that is constant. Since each $X_i$ can only depend on its own $U_i$ and possibly the other $ \bX$, one of the $X_i$ is constant given the other one. Since the $\bX$ are binary, we have 
$\countPosX{\fM'} \le 2$ and any small Markovian model cannot be equivalent to $\fM$. 
\end{proof}

\begin{proof}[Proof of Lemma~\ref{lemma:small:model:in:linear}]
    If one examines the proofs in \citep{doerflerICLR2025}, all models have the small-model property.
    For the proofs of $\SATprobcompsumsm$,  $\SATproblinsumsm$,  $\SATinterventcompsumsm$, and $\SATinterventlinsumsm$, a small-model is explicitly mentioned as the proof apply a lemma of \citep{fagin1990logic} to find a solution in the rational numbers. 
    The proofs for $\SATcausalcompsumsm$ and $\SATcausallinsumsm$ construct a purely deterministic model using the functions $\cF$, which works with any distribution $P$ of the unobserved variables $\bU$ and thus also with a small one. 
\end{proof}

We prove the $\NEXP$-completeness of the interventional
and counterfactual satisfiability problems with small model property as follows.

\begin{proof}[Proof of Theorem~\ref{lab:lemma:interv:counter:smconj}]
    We show that $\SATinterventpolysumsm$ is $\NEXP$ hard.
    This is done similarly as in the proof of Theorem~\ref{thm:SATinterventcompsumgraph:ccNEXP-complete}, again reducing from the satisfiability of Schönfinkel-Bernay sentences.
    The setup of all random variables is identical, however, this time we have to use the polynomial arithmetic to encode our constraints.
    
    Using Lemma~\ref{lemm:causal:ordering}, we require that each model has the variable ordering $X_1 \prec \ldots \prec X_n \prec Y_1 \prec \ldots \prec Y_n$ and then the variables $R_i$ and $R_i^j$ in some arbitrary, but fixed order.
    This allows to enforce a deterministic model by adding the constraint $\sum_{\bt} \sum_{\bv} (\PP{[\bt]\bv}^2 - \PP{[\bt]\bv})^2 = 0$ for each subset $T$ of variables containing a (possibly empty) consecutive subset of variables starting at $X_1$ in our variable ordering.
    This enforces every probability in this model to be either $0$ or $1$, and it further enforces that no variable depends on any exogenous variables but only on the endogenous variables.
    This determinism allows us to rephrase the additional equations from \cite{doerflerICLR2025} using purely interventional terms:
    
    We use the following constraint to ensure that $R_i$ only depends on its arguments: 
    \begin{equation}
        \sum_{\bv}(\PP{[z^1_i,\ldots,z^k_i] R_i=r_i} - \PP{[\bv\setminus r_i] R_i=r_i})^2 = 0
    \end{equation}
    
    Thereby $\sum_{\bv}$ refers to summing over all values of all endogenous variables in the model and the constraint says that an  intervention on $Z^1_i,\ldots,Z^k_i$ gives the same result for $R_i$ as an intervention on $Z^1_i,\ldots,Z^k_i$ and the remaining variables, excluding $R_i$.
    
    We use the following constraint to ensure that $R^j_i$ only depends on its arguments: 
    \begin{equation}\label{eqn:proof:interv:sm:nexptime:relation:occurrence}
        \sum_{\bv} (\PP{[t_1,\ldots,t_k] R^j_i=r^j_i} - \PP{[\bv\setminus r^j_i] R^j_i=r^j_i})^2 = 0
    \end{equation}
    
    and that $R^j_i$ and $R_i$ have an equal value for equal arguments:
    \begin{equation}\label{eqn:proof:interv:sm:nexptime:consistent:relations}
        \sum_{t_1,\ldots,t_k} (\PP{[T_1=t_1,\ldots,T_k=t_k] R^j_i=r_i} - \PP{[Z^1_i = t_1,\ldots,Z^k_i=t_k] R_i=r_i})^2 = 0\,.
    \end{equation}
    
    We add the following constraint for each $X_i$  to ensure that the values of $\bX$ are not affected by  the values of $\bY$:
    \begin{equation}\label{eqn:proof:interv:sm:nexptime:x:y:order}
        \sum_{\bv}\sum_{\by'} (\PP{ [\bv \setminus x_i ] X_i=x_i} - \PP{[\bv \setminus (x_i, \by),\bY=\by'] X_i=x_i})^2 = 0 
    \end{equation}
     
    Let $\psi'$ be obtained from $\psi$ by replacing equality and relations on the Boolean values with the corresponding definitions for the random variables:
    \begin{equation}\label{eqn:proof:interv:sm:nexptime:main:equation}
        \sum_{\by} \PP{[\by] \psi' } = 2^n
    \end{equation}

    The correctness now follows by the same argument as Theorem~\ref{thm:SATinterventcompsumgraph:ccNEXP-complete}.
    On one hand, the model $\fM$ constructed there if $\exists \bx \forall \by \psi$ is satisfiable is already deterministic and thus fulfills all the above constraints.
    On the other hand, if there is a model $\fM$ satisfying the above constraints, all possible values $\bu$ of the exogenous variables result in the exact same probabilities which, following the original proof, lead to $\exists \bx \forall \by \psi$ to be satisfiable.

    Furthermore $\SATcausalpolysumsm \in \NEXP$.
    We adapt a trick from \cite[Thm. 14]{doerflerICLR2025}.
    Instead of viewing the model as changing with $\bu$, we instead consider for each value $\bu$ a separate (deterministic) model.
    Due to the small-model property, we then only have to consider polynomially many models and each model has a straight-forward representation using exponentially many bits, directly encoding all the functions in $\cF$ as explicit tables.
    Algorithm~\ref{alg:causal:polysm:sum:in:NEXP} uses this observation to solve $\SATcausalpolysumsm$ non-deterministically in exponential time.
    It constructs an $\ETR$-formula with polynomially many variables, polynomially many polynomials of polynomial degree, but potentially exponentially many monomials.
    Each coefficient of the formula can be represented using a polynomial bit length.
    Renegar's algorithm \cite{renegar1992computational} can thus solve the $\ETR$-formula deterministically in time $O(\poly(|\varphi|)^{O(p)})$, which is exponential in the input length.
\end{proof}

\end{document}